\documentclass[12pt]{scrartcl}
\usepackage[utf8]{inputenc}
\usepackage[english]{babel}
\usepackage{amsmath}
\usepackage{amsfonts}
\usepackage{amssymb}
\usepackage{graphicx}
\usepackage{fancyhdr}
\usepackage{units}
\usepackage{icomma}
\usepackage{epstopdf}
\usepackage{isotope}
\usepackage{cite}
\usepackage{hyperref}
\usepackage{bm}
\usepackage{caption}
\usepackage{subcaption}
\usepackage{sidecap}
\usepackage{listings}
\usepackage{wrapfig}
\usepackage{url}
\usepackage{authblk}
\usepackage{float}
\usepackage{booktabs}
\usepackage[toc,page]{appendix}
\usepackage{tensor}
\usepackage{cancel}
\usepackage{euscript} 
\usepackage{ mathrsfs }
\usepackage{amsthm}
\usepackage{nameref}
\usepackage{tikz}
\usepackage{wrapfig}
\usepackage{multirow}
\usepackage{array}
\usepackage{lipsum}
\usepackage{cleveref}
\usepackage{xcolor, soul}
\usepackage{makecell}
\usepackage{ textcomp }
\sethlcolor{green}
\usetikzlibrary{positioning}

\usepackage[miktex]{gnuplottex}
\usepackage{latexsym}
\usepackage{keyval}
\usepackage{ifthen}
\usepackage{moreverb}
\usepackage{setspace}
\newtheorem{theorem}{Theorem}[section]

\def\disp{\displaystyle}
\def\be{\begin{equation}}
\def\ee{\end{equation}}
\def\beq{ \begin{eqnarray}}
\def\eeq{ \end{eqnarray}}
\def\bes{\begin{equation*}}
\def\ees{\end{equation*}}
\def\const {{\textrm const}}
\def\dd {{\textrm d}}
\def\d{{\textrm d}}
\def\RR{{\mathbb R}}

\title{Retirement spending problem under Habit Formation Model}
\author[1]{S. Kirusheva\footnote{skir@yorku.ca. Kirusheva's research was supported in part by MITACS and CANNEX Inc.}}
\author[2]{H. Huang}
\author[3]{T.S. Salisbury\footnote{salt@yorku.ca.  Salisbury's and Huang's research were suppoted in part by grants from NSERC.}}
\affil[1,2,3]{Department of Mathematics and Statistics, York University, Toronto, Ontario, Canada}
\affil[2]{Research Centre for Mathematics, Advanced Institute of Natural Sciences, Beijing Normal University, Zhuhai, Guangdong, China}
\affil[2]{BNU-HKBU United International College, Zhuhai, Guangdong, China}
\date{}
\begin{document}
\maketitle

\pagenumbering{arabic}
\pagestyle{myheadings}
\vspace{-1cm}
\begin{center}
Abstract
\end{center}
\vspace{-0.5cm}
\begin{abstract} 
In this paper we consider the problem of optimizing lifetime consumption under a habit formation model. Our work differs from previous results, because we incorporate mortality and pension income. Lifetime utility of consumption makes the problem time inhomogeneous, because of the effect of ageing. Considering habit formation means increasing the dimension of the stochastic control problem, because one must track smoothed-consumption using an additional variable, habit $\bar c$. Including exogenous pension income $\pi$ means that we cannot rely on a kind of scaling transformation to reduce the dimension of the problem as in earlier work, therefore we solve it numerically, using a finite difference scheme. We also explore how consumption changes over time based on habit if the retiree follows the optimal strategy. Finally,  we answer the question of whether it is reasonable to annuitize wealth at the time of retirement or not by varying parameters, such as asset allocation $\theta$ and the smoothing factor $\eta$.
\end{abstract}
\vspace{-0.8cm}
%-------------------------------------------------
\section{Introduction}
%-------------------------------------------------
\subsection{Motivation}
%-------------------------------------------------
Nowadays, we observe a growing interest in investment plans that give a potential client the confidence of a stable income over the course of retirement.  The main goal of any such plan is to find the strategy that minimizes the risk of ruin and, at the same time, maximizes the level of consumption.  Our current research deals with a retirement spending problem (RSP) under dynamics that include the individual's  habit. In other words, we take into consideration how much a retiree usually spends, i.e. we solve the problem under a habit formation model (HFM). This postulate makes the model much more difficult to solve. 

In this paper our goal is to explore how the presence of exogenous income in the model that includes the client's habit will affect the optimal consumption and compare these results with different models, such as HFM without pension for two cases, with and without asset allocation. Also we answer the question of how much the agent can spend during retirement based on his initial  wealth $w$ and consumption habit $\bar c.$ 

There is one more interesting option for a retiree that we also discuss in this article. The individual can convert some or all of his initial wealth into annuities. We briefly discuss this possibility, and if it is reasonble to do so at age of $65$  depending on the client's habit. We assume that once he converts his wealth into annuities he can't reverse the transaction. We obtain numerical results for different parameters, such as habit $\bar c,$ asset allocation $\theta,$ and smoothing factor $\eta$.
  
%-------------------------------------------------
\subsection{Literature review}
%-------------------------------------------------
Many articles have been written on this topic that consider various scenarios. Lately more researchers are paying more attention to the HFM  when they deal with financial questions.
There are several articles that solve similar problems, somehow related to the habit formation model, for example \cite{B}, \cite{CA}, \cite{CH}, \cite{N}, \cite{P}, \cite{Pollak}. All of them use  different approaches and techniques. First,  we should mention one of the earliest papers \cite{C}. In that article the author tried to solve the equity premium puzzle (EPP) which was first formalized in a study by Rajnish Mehra and Edward C. Prescott \cite{MP}. Under the assumption of rational expectations this problem was resolved. One of the issues with that formulation is that the consumption should be always greater or equal to the exponentially weighted average (EWA) of consumption which is hard to implement in real situations. As a consequence, there are modifications of this work which are discussed, for instance in the book \cite{R} which introduces a novel form for the HFM utility. Another attempt to resolve the EPP was made by the authors \cite{X}. They considered optimal portfolio and consumption selection problems  with habit formation in a jump diffusion incomplete market in continuous-time. One more pioneering work \cite{Pollak} describes a model of consumer behaviour based on a specific class of utility functions, the so-called ``modified Bergson family''. 

Another example of using HFM was covered in the following article \cite{P} which explores the implications of additive and endogenous habit formation preferences in the context of a life-cycle model of consumption and portfolio choice for an investor who has stochastic uninsurable labor income. In order to get a solution he derives analytically constraints for habit and wealth and explains the relationship between  the worst possible path of future labor income and the habit strength parameter. He concludes that  even a small possibility of a very low income implies more conservative portfolios and higher savings rates. The main implications of the model are robust to income smoothing through borrowing or flexible labor supply.

Finally, we would like to mention one more paper \cite{Jer} where the authors proved existence of optimal consumption-portfolio policies for utility functions for which the marginal cost of consumption (MCC) interacted with the habit formation process and satisfied a recursive integral equation with a forward functional Lipschitz integrand and for  utilities for which the MCC is independent of the standard of living and satisfied a recursive integral equation with locally Lipschitz integrand.

There are a lot of financial strategies which can help to plan how to spend money under different preferences but one of the most important targets is to arrange consumption during retirement. There are many works devoted to retirement spending plans, such as \cite{B}, \cite{H},\cite{HH}, \cite{MHH}. For example, in the paper \cite{B} the authors discuss consumption and investment decisions in a life-cycle model under a habit formation model incorporating stochastic wages and labor supply flexibility. One of the results shown was that utilities that exhibit habit formation and consumption-leisure complimentarities induce an impact of past wages on the consumption of retirees. Hence the authors  showed that it is important to take into consideration  habit and consumption-leisure complimentarities when formulating life-cycle investment plans. In the next article \cite{H} the authors consider a model based on results from the article \cite{MHH} where a similar problem was solved under assumption of deterministic investment returns. In \cite{H} the authors accept stochastic returns and then compare optimal spending rates with the analytic approach from the article \cite{MHH}.  When a potential client starts to think about a retirement spending plan there is one more question that arises, namely under which conditions he can consider investment into annuities for part or all of his wealth. To be precise, when we say ``annuities'' we mean life annuities,  insurance products that pay out a periodic amount for as long as the annuitant is alive, in exchange for a premium (see \cite{Br}). This question has been widely discussed in the literature, for example \cite{Mil}, \cite{MH} or \cite{Rei}.

 In the recent article \cite{H} RSP was solved for fixed risky asset allocation $\theta=\const.$ Here we solve a similar problem following HFM, using the novel utility of \cite{R}. 
%-------------------------------------------------
\subsection{Paper Agenda}
%-------------------------------------------------
The paper is organized as follows. In Section \ref{subsection_1} we explain what the habit formation model is and formulate our problem for two different cases,  without  (see Section \ref{subsection_1cs}) and with (see Section \ref{subsection_3})  pension. In Section \ref{numerical_results} we discuss how the smoothing factor $\eta$ affects the numerical solution and  provide a comparison between two different cases, without pension (Section \ref{wo_pension}) and  when the client has constant pension income (Section \ref{comparison1} and \ref{comparison2}).

Section \ref{wealth_sim} is devoted to discussing how the client can spend money during his retirement based on a given initial amount of wealth $w$ and a certain habit $\bar c.$
In the Section \ref{wealth} we analyze the possibility of annuitizing wealth, entire or partially, at the age of $65.$ 

As with most of these models, this one doesn't have an analytical solution and has to be approximated numerically. Many algorithms have been developed through the years. Every algorithm has its own advantages and disadvantages, which differ in accuracy and efficiency.  In this paper we chose a finite difference scheme for its simplicity and accuracy. The detailed description of the approximation  scheme, some error analysis as well as some theoretical background are provided in the Appendices (see  \ref{appendix_theorem}, \ref{appendix_FDS}, \ref{error}). In Appendix  \ref{numerical_res} we summarize numerical results obtained for different sets of wealth $w$ and habit $\bar c$. Finally, in Appendix \ref{section_comparison}, we describe some numerical results for different sets of asset allocations $\theta$ and volatility $\sigma$ for fixed smoothing factor $\eta=1.0,$ as the most interesting case for solving problem under HFM.

%-------------------------------------------------
\section{Model formulation}
\label{subsection_1}
%-------------------------------------------------
\subsection{Introduction}
\label{theor_backgr}
When we think about the model we should think about a client, more precisely about a retiree, who has a certain amount of wealth $w$ and who wants to know what to do next with his endowment. In order to decide how much he can spend we should understand how wealth is changing over the time counting all possible income and expenses, as in the following:
\beq
\begin{array}{c}
\dd w_t=[\theta(\mu-r)+r]w_t \dd t+\theta \sigma w_t  \dd W_t+\pi \dd t-c_t \dd t \\
\dd\bar c_t=\eta (c_t-\bar c_t)\dd t.
\end{array}
\label{TB0}
\eeq
We can provide the following explanation of the equations (\ref{TB0}). There is a part of wealth $w$ which grows at the riskless rate $r$, there is a stochastic component  represented by the parameter $\theta$,  which is the fraction of wealth invested into risky assets (in our case, we take it as a fixed parameter $\theta$), drift  $\mu$, volatility $\sigma$  and $W_t$ a Brownian Motion (BM).  Also assume that there is an exogenous fixed income, pension $\pi.$ We solve our problem using a habit formation model. This means that the agent's  utility depends on the consumption rate  $c_t$ and on an EWA $\bar c_t$ of consumption rates over previous time periods. We will consider the finite-horizon problem therefore the client's objective function is taken to be
\beq
V(t,w,\bar c)=\sup_{c_t} E\left[\int_t^T e^{-\rho s}\tensor[_s]{p}{_x}u\left(\frac{c_s}{\bar c_s}\right)\d s|w_t=w,\bar c_t=\bar c\right]
\label{TB5}
\eeq
where $\rho$ is the personal time preference or subjective discount rate,  $\tensor[_s]{p}{_x}$ is the probability of survival from the retirement age $x$ to $x+s.$
We set up  the probability of survival based on the Gompertz Law of Mortality (\cite{M}), i.e.
\be
\tensor[_s]{p}{_x}=e^{-\int_0^t \lambda_{x+q}\d q}.
\label{3}
\ee
Here $\lambda$ is the biological hazard rate
$
\lambda_{x+q}=\frac{1}{b}e^{(x+q-m)/b}
$
where $m$ is the modal value of life (see p47, \cite{M}), $b$ is the dispersion coefficient of the future lifetime random variable,
$u$ is the CRRA  utility function
\be
u(c)=\frac{c^{1-\gamma}}{1-\gamma}
\label{4}
\ee
where $\gamma$ is the risk aversion parameter. Note that the formulation of utility in (\ref{TB5}) is due to  \cite{R} and differs from much of the HFM literature. This  changes the nature of the solutions.

In order to solve this problem, we will use the value function approach that implies we should derive an Hamilton-Jacobi-Bellman (HJB) equation for our model. 
Let us consider two different  cases. The problem that we are going to solve first, reflects an agent's expectations who doesn't have any pension income which means that his wealth's growth results partly from investing in a bank account growing at the risk free rate and partly from investing  another portion of the wealth into risky assets. The second problem involves the presence of pension income. In addition, the asset allocation will be fixed.
%----------------------------------------------
\subsection{Model without pension.}
\label{subsection_1cs} 
%----------------------------------------------
In this paragraph we consider wealth and habit dynamics (\ref{TB0}) without pension and with asset allocation as a parameter $\theta=const$ with the client's objective function described as follows (see Appendix \ref{appendix_theorem}(\ref{verification_})). Since we use the value function approach,  we need to derive an HJB equation using a verification theorem (\ref{verification_theorem}) but, first, let's change variables. This will allow us to reduce the dimension of our problem by analogy with one introduced in the book \cite{R}
\beq
\begin{array}{c}
\vspace{0.2cm}
x_t\equiv \disp\frac{w_t}{\bar c_t},\quad q_t=\disp\frac{c_t}{\bar c_t},  \\
V(t,w,\bar c)=V\left(t,\disp\frac{w}{\bar c},1\right)\equiv \nu\left(t,\disp\frac{w}{\bar c}\right)=\nu(t,x)
\end{array}
\label{eq4}
\eeq
where $V(t,w,\bar c)$ is the unknown value function defined by equation (\ref{TB5}).
Then the dynamics of the new scaled wealth $x$ will be the following
\beq
\dd x_t= \dd\left(\frac{w_t}{\bar c_t}\right)=rx_t\dd t+\theta((\mu-r)x_t\dd t+\sigma x_t \dd W_t)-(\eta x_t+1)q_t\dd t+\eta x_t \dd t. \nonumber 
\eeq
Omitting all calculations  we provide only the final HJB equation using the new variables
\be
\nu_t-(\rho+\lambda) \nu+\alpha\nu_x + \beta\nu_{xx}+u(q^*)=0
\label{HJB_2}
\ee
where the optimal consumption $q$ will be 
\beq
q^*=[(\eta x+1)\nu_x]^{-\frac{1}{\gamma}}.
\eeq
and the coefficients $\alpha=\alpha(x)$ and $\beta=\beta(x)$ in the formula have the following  form 
\newline $\alpha= (\theta(\mu-r)+r+\eta)x-(\eta x+1)q^*, \quad \beta=\frac{1}{2}\theta^{2}\sigma^2x^2. $

Boundary conditions for this problem are taken to be following: At the boundary $x=0$ assume that optimal consumption is $q^*_t=0$ since there is not any other income.
Then at the other boundary where $x=x_{\max}$ assume that the value function gradually approaches zero which implies that the value function derivative is zero, i.e. $\nu_x=0.$

%----------------------------------------
\subsection{Model with pension.}
\label{subsection_3}
%----------------------------------------
In this paragraph we expand on the previous problem and now, assume that the client not only invests part of his wealth into a bank account and makes profits from stocks, but also has  pension income. 

Now, suppose that the wealth dynamics and objective function can be described by equations (\ref{TB0})- (\ref{TB5}).   The goal is to maximize the value function by controlling the consumption $c_t$ such that the wealth $w_t$ remains non-negative with fixed asset allocation. In Appendix \ref{appendix_theorem} \ref{subsection_1} we derive and prove a verification theorem for the following HJB equation
\beq
\hspace{-1cm}\sup_{c_t}\Bigl[\tilde V_t+\tilde V_w(\theta(\mu-r)+r)w_t+\tilde V_w(\pi-c_t)+ \tilde V_{\bar c}\eta (c_t-\bar c_t)+ \Bigr. \nonumber \\
\Bigl.\frac{1}{2}\tilde V_{ww}\theta^2\sigma^2w^2_t+e^{-\rho t}\tensor[_t]{p}{_x} u\left(\frac{c_t}{\bar c_t}\right)\Bigr]=0. 
\label{HJB}
\eeq
where the consumption rate $c_t$ is the only our control variable.  Now, we introduce new notation for our convenience
\be
\tilde V(t,w_t,\bar c_t)=e^{-\rho t}\tensor[_t]{p}{_x} V(t,w_t,\bar c_t).
\ee

By performing standard calculations we can obtain the HJB  equation
\beq
 V_t-(\rho+\lambda_{t+x})V+ V_w( (\theta(\mu-r)+r)w+\pi-c^*)+\nonumber \\
  V_{\bar c}\eta (c^*-\bar c)+  \frac{1}{2}V_{ww}\theta^2\sigma^2w^2+u\left(\frac{c^*}{\bar c}\right)=0.
\label{PDE_c_opt}
\eeq
where the optimal consumption has the form
\be
c^*=\bar c^{\frac{\gamma-1}{\gamma}} (V_w-V_{\bar c}\eta)^{-\frac{1}{\gamma}}.
\label{opt_c}
\ee

Since we deal with a nonlinear second order PDE (\ref{PDE_c_opt}) we should set up boundary conditions.
At the terminal time $T$ the integral in the formula that represents the value function $V(t,w,\bar c)$ is zero. 
At the boundary with zero wealth, $w=0,$ we impose the constraint $c^*<\pi.$ Then we will get a simple first-order PDE 
\beq
 V_t-(\rho+\lambda_{t+x})V+  V_w( \pi-c^*)+V_{\bar c}\eta (c^*-\bar c) + u\left(\frac{c^*}{\bar c}\right)=0
 \label{BC_1}
\eeq

If we assume that the value function is asymptotically proportional to wealth
\be
%V(t,w,\bar c) \sim  f(t,\bar c)w^{1-\gamma} \quad {\rm then} \quad V_w \sim f(t,\bar c)w^{-\gamma}
V(t,w,\bar c) \sim  f(t,\bar c)w^{1-\gamma} \quad {\textrm then} \quad V_w \sim f(t,\bar c)w^{-\gamma}
\ee
where $f(t,\bar c)$ is some arbitrary function of time and the EWA of consumption. When the wealth is big enough, i.e. $w\to \infty$ the derivative of the value function goes to zero asymptotically
\be
V_w=0.
\label{BC_2}
\ee
This means that changes in wealth $w$ are not that important and do not affect the utility function as much as in the case when $w\to 0.$

The detailed discussion of the approximation scheme and discretization error is provided in the Appendices (see \ref{appendix_FDS} and \ref{error}).
%----------------------------------------------
\section{Smoothing factor effect on numerical solution}
\label{numerical_results}
%----------------------------------------------
%----------------------------------------------
\subsection{HFM w\textbackslash o pension.}
\label{wo_pension}
%---------------------------------------------- 
We start our discussion by presenting results where the retiree does not have any pension income. Since the problem does not have an analytical solution we have to solve it numerically. For that we choose an implicit upwind method described in detail in Appendix  \ref{appendix_FDS}. 
Below we provide some numerical results for two cases. In the first picture (see Figure (\ref{test1}))  we see the numerical results for fixed asset allocation $\theta=0.6,$ whereas in the second picture (see Figure \ref{test2}) we solved the optimization problem (OP) where asset allocation is a control variable $\theta_t$ (this is close to what is done in \cite{R}).  In order to solve this problem we also set some parameters, namely:
risk-free rate $r=0.02,$ drift $\mu=0.08,$ risk aversion parameter $\gamma=3$ and volatility $\sigma=0.16.$ Figure (\ref{test}) shows the solution with these parameters, at age $65.$
\begin{figure}[H]\centering
\begin{subfigure}[c]{0.4\textwidth}
\includegraphics[scale=0.3]{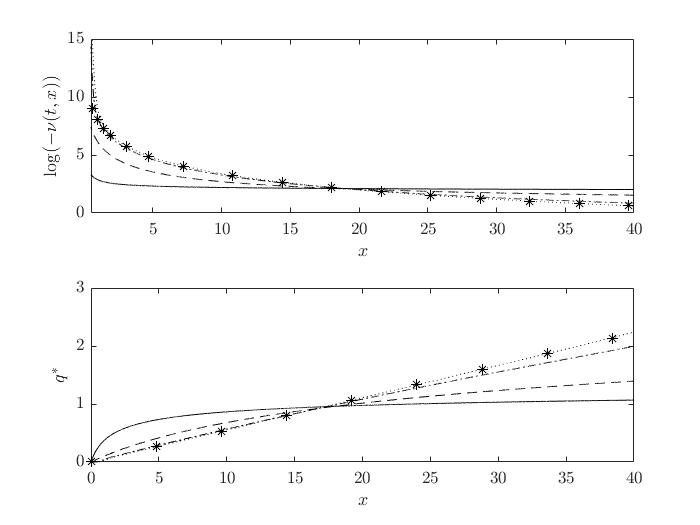}
\subcaption{$\theta=0.6$ is fixed.}
\label{test1}
\end{subfigure}
\hspace{1cm}
\begin{subfigure}[c]{0.4\textwidth}
\includegraphics[scale=0.3]{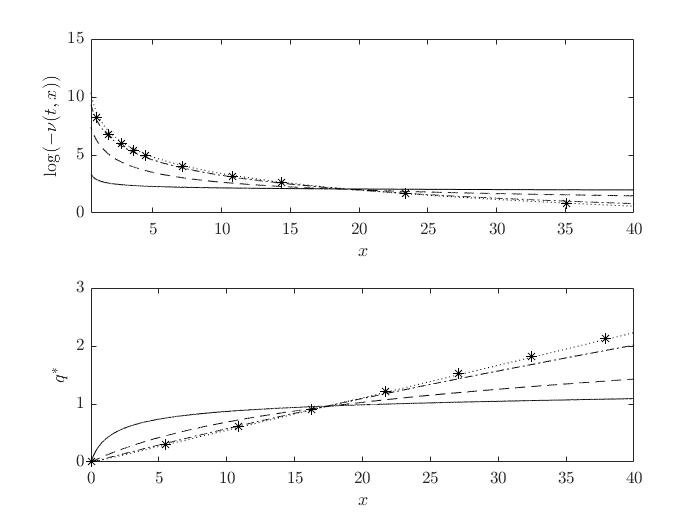}
\subcaption{$\theta$ is a control variable.}
\label{test2}
\end{subfigure}
\caption{Solution for OP w\textbackslash o pension for asset allocation $\theta$ in \ref{test1} as a fixed parameter $\theta=0.6$ and as a control variable in \ref{test2} where the solid line represents $\eta=1.0,$
the dashed line - $\eta=0.1,$ the dash-dot line - $\eta=0.01$ and the dotted line - $\eta=0$ and $*-$ line represents an alternative numerical solution for the last case at the present time.\vspace{1cm}}
\label{test}
\end{figure}
In Figure (\ref{test1}) the upper picture shows the relationship between the logarithm of the negative value function ($\log(-\nu(t,x))$) and the logarithm of the ratio of wealth to habit ($\log(w/\bar c)$) for different values of the smoothing factor $\eta=0,\;0.1\;,0.01,\; 1.0.$ The solid line represents $\eta=1.0,$ the dashed line - $\eta=0.1$ and the dash-dot line represents $\eta=0.01.$ The last case corresponds to the Merton problem $\eta=0$ (dotted line). The star line represents another numerical solution for the later case with the assumption that value function has the form  $\nu(t,x)=h(t)u(x) $  where the function $h(t)$ satisfies the following ODE
\begin{equation}
h'+(-\rho +(1-\gamma)(r+\theta(\mu-r)-\frac{\gamma}{2}\theta^2\sigma^2))h+\gamma h^{\frac{\gamma-1}{\gamma}}=0.
\label{bernulli}
\end{equation}
Then the solution will have the following form and also can be solved numerically
\beq
f_1=-\rho+(1-\gamma)\left(\theta(\mu-r)+r-\frac{\gamma}{2}(\theta\sigma)^2\right),  \qquad 
f_2=e^\frac{x-m+T}{b}, \nonumber \\
h(t)=\left(e^{{\frac{1}{\gamma}\left(f_1(T-t)-f_2+e^{(x-m+t)/b}\right)}}\left(1+\int_t^Te^{-\frac{1}{\gamma}\left(f_1(T-t^\prime)-f_2+e^{(x-m+t^\prime)/b}\right)}\dd t^\prime\right)\right)^\gamma.
\eeq
If we look at the upper pictures (see Figure \ref{test1} or \ref{test2})  they show that in the beginning the value function decreases dramatically and then gradually slows down and settles into asymptotic behaviour.  This reflects the intuition about what the value function is. When wealth is small then its changes have a great impact on the value function. As soon as changes in wealth become significantly smaller than an absolute value of the wealth the value function slows down.
The bottom picture in Figure (\ref{test1}) represents the relationship between the optimal ratios of consumption-habit ($c/\bar c$) and wealth-habit ($w/\bar c$) for four values of the parameter $\eta=0,\; 0.01,\; 0.1,\; 1.0$. As we can see in the case with $\eta=1.0$ (solid line), the possibility for consumption $c$ to exceed habit exists only in case of very  big wealth,  for other values the rule is following: the less responsive the consumption is to changes in habit, the higher the optimal consumption will be. 

Figure \ref{test2}  represents two graphs where $\theta_t$ is a control variable. All calculations can be done analogously with the only difference being that now, we need  find an optimal value for asset allocation as well. It can be computed as follows
\beq
\theta^*=-\frac{\kappa}{\sigma}\frac{\nu^\prime}{\nu^{\prime\prime}x}, \qquad \kappa=\frac{\mu-r}{\sigma}.
\eeq

We can also see the analogous behaviour for both functions, the value function (upper picture) and the optimal consumption (lower picture). Results also were obtained for  four values of the smoothing factor $\eta=0,\; 0.1,\; 0.01,\; 1.0.$ Alternative solution for the case $\eta=0$ also can be computed by analogy with previous case. Similar results for the smoothing parameter $\eta=1.0$ were obtained in the book \cite{R}.
%---------------------------------------------- 
\subsection{Comparison two models, with and  w\textbackslash o pension.}
%---------------------------------------------- 
In this section we illustrate how changes in the smoothing factor $\eta$ will affect our numerical solution. Furthermore, in order to better illustrate how the solutions differ, we compare results with and without pension (see \ref{wo_pension}). 
%---------------------------------------------- 
\subsubsection{ Comparison based on the coordinate transformation $ (t,w,\bar c)\mapsto(t,x).$ }
\label{comparison1}
%---------------------------------------------- 
In this case we scale our results with pension and compare them with those from the previous paragraph (Section \ref{wo_pension}). We provide results below (see Figure (\ref{f26})) for our optimization problem where the solid lines represent how results with pension converge to those without one. Different solid lines represent particular habit values $\bar c.$ 
\begin{figure}[H]
\begin{minipage}[b]{.5\textwidth}
\includegraphics[scale=.27]{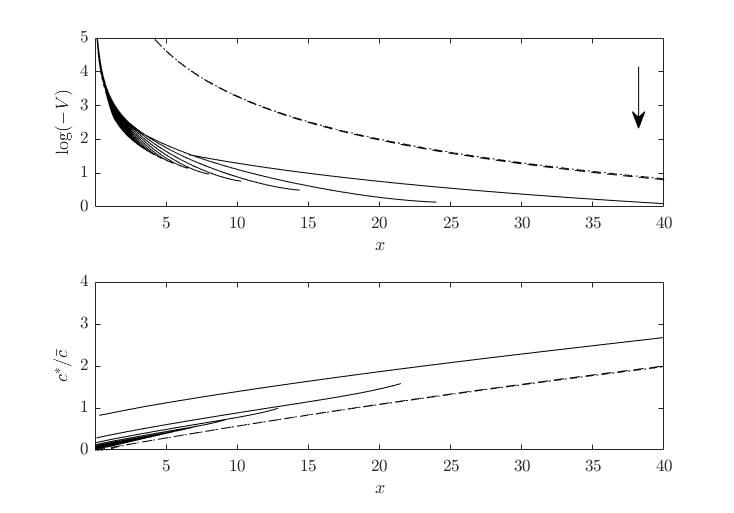}
\subcaption{$\eta=0.01$ }
\label{f26_4}
\end{minipage}
\begin{minipage}[b]{.5\textwidth}
\includegraphics[scale=.27]{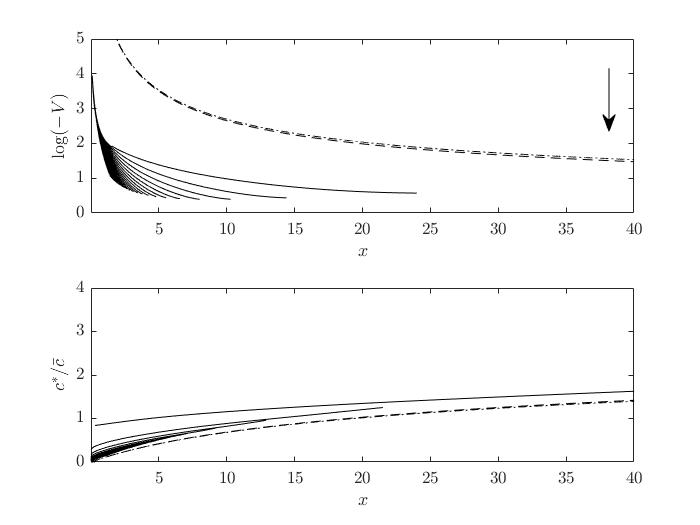}
\subcaption{$\eta=0.1$ }
\label{f26_3}
\end{minipage}
\begin{minipage}[b]{1.0\textwidth}
\centering
\includegraphics[scale=.27]{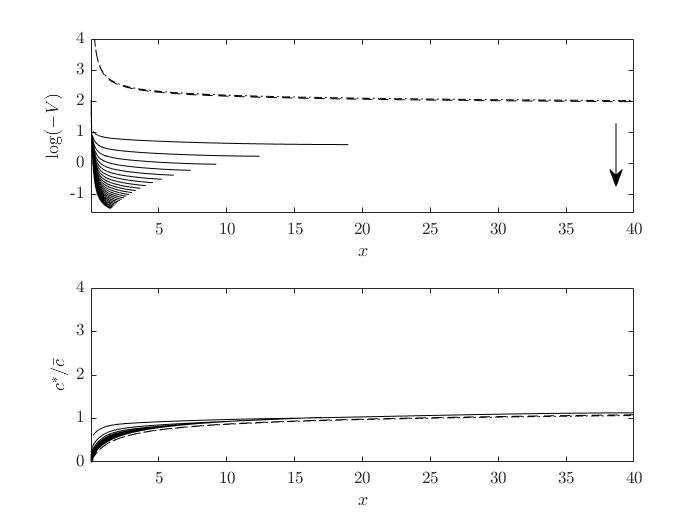}
\subcaption{$\eta=1.0$ }
\label{f26_5}
\end{minipage}
\caption{Scaled value function $V$ (upper picture) and scaled optimal consumption $c^*\slash \bar c$ (bottom picture) vs. scaled weath $x$ for two cases  w\textbackslash o pension, with (dashed line)  and w\textbackslash o (dash-dot line) asset allocation. In the latter, $\theta=0.6.$ The solid lines represent the value function $V$ and scaled optimal consumption $c^*\slash \bar c$ for varioius values of habit $\bar c$.}
\label{f26}
\end{figure}

There are three sets of graphs for three values of the smoothing factor $\eta.$ Every set consists of two pictures where the upper one represents the relationship between the logarithm of the scaled value function $\log(-V)$ and scaled wealth $x.$ The arrow shows the direction of increasing habit and the second graph represents scaled optimal consumption $c^*\slash \bar c$ and scaled wealth $x$. In the cases with pension, we show $V$ and $c^*\slash \bar c$ in place of $\nu(t,x)$ and $q^*,$ for various values of habit $\bar c.$ The dashed line represents a problem where asset allocation is a control variable which proves empirically that an optimal strategy is better then any arbitrary strategy. In our case the dash-dot line represents a fixed parameter $\theta$ (see the bottom pictures on Figure (\ref{f26})). These curves are very close, for reasons we will explain shortly.  Parameters are as in Figure (\ref{test}). 

Results were obtained for three values of the smoothing factor $\eta=10^{-2}$ (see Figure (\ref{f26_4})), $\eta=10^{-1}$ (see Figure (\ref{f26_3})) and $1.0$ (see Figure (\ref{f26_5})). The first picture represents the case where $\eta=0.01$ (Figure \ref{f26_4}). In this case we can see that curves that represent consumption  approach to each other very slow because of great impact of averaging. For the next  value of the smoothing factor, $\eta=0.1,$ consumption $c^*\slash \bar c$ adapts to the habit $\bar c$ faster (Figure \ref{f26_3}). Here we see a smaller gap between the solutions with (solid line) and without pension (dashed and dash-dot lines). This happens because the smoothing factor is significantly bigger which implies less impact from averaging and, as a consequence, faster response to changes in habit. In the last case when $\eta=1.0$ (Figure \ref{f26_5}) we clearly see that changes happen very fast. For relatively big values of scaled wealth $x$ solutions are very close. In order to see indistinguishable results, we would have to consider a bigger scale which is computationally expensive.

We can think about the case when our habit is the same order as pension, $\bar c \approx \pi$ or equivalently $x\approx w/\pi.$ %This case on the graph represented as a ``star-line''. 
If this ratio is small we would expect to see that our numerical results differ from those that represent the case w\textbackslash o pension. At the same time, as $x$ increases, curves that represent the optimal consumption for all three cases approach to each other (see Figure \ref{f26_5}). If we consider an increasing habit $\bar c$ then the difference between our numerical results with and without pension will increase. On the graph habit increases in the downward direction. The next picture shows slightly different behaviour of our numerical solution (see Figure (\ref{f26_4})).  Overall, we can say that consumption with pension income is definitely greater than without pension and it depends on habit and how fast the model reacts to changes, in other words how big the smoothing factor $\eta$ is. 

%---------------------------------------------- 
\subsubsection{ Comparison based on coordinate transformation $(t,x)\mapsto (t,w,\bar c).$ }
\label{comparison2}
%----------------------------------------------
In the previous paragraph \ref{comparison1} we discussed results for scaled values of optimal consumption $q^*$ and wealth $x.$  Here we will consider the opposite transformation, $(t,x)\mapsto (t,w,\bar c).$ 
\begin{figure}[H]
\begin{subfigure}{.5\linewidth}
\centering
\includegraphics[scale=0.27]{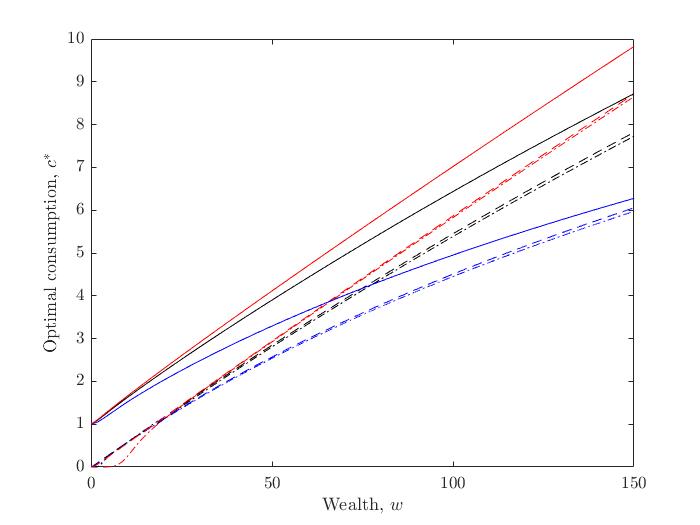}
\caption{$\eta=0.01$ }
\label{f25_3}
\end{subfigure}
\begin{subfigure}{.5\linewidth}
\centering
\includegraphics[scale=0.27]{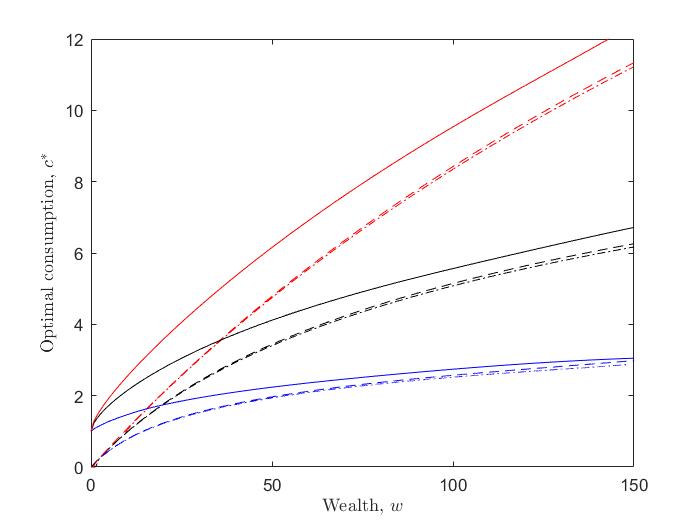} 
\caption{$\eta=0.1$ }
\label{f25_2}
\end{subfigure}\\[1ex]
\begin{subfigure}{1\linewidth}\centering
\includegraphics[scale=0.27]{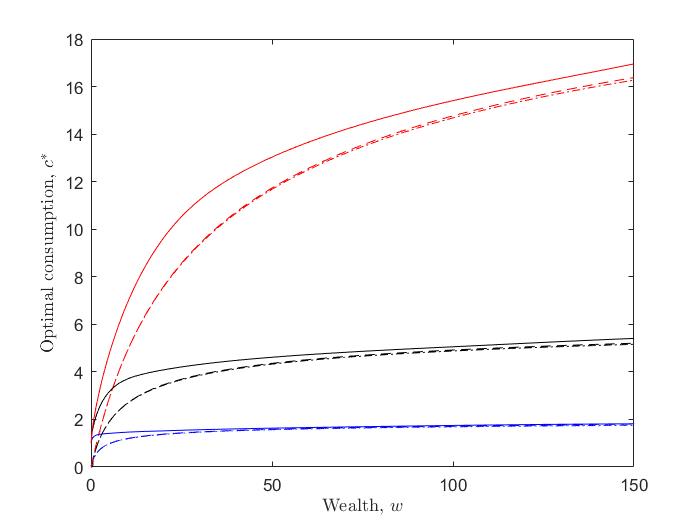}
\caption{$\eta=1.0$ }
\label{f25_1}
\end{subfigure}
\caption{Optimal consumption $c^*(t,w,\bar c)$ vs. weath $w$ for multiple values of habit $\bar c.$ Comparison between results with pension $\pi$ for a fixed parameter $\theta$ (solid line), w\textbackslash o pension and w\textbackslash o asset allocation (dash-dot line) and w\textbackslash o pension and with $\theta_t$ as a control variable (dashed line). The habit is fixed at $\bar c=1$ (blue lines), $5$ (black lines) and $20$ (red lines).}
\label{f25}
\end{figure}

In Figure \ref{f25} we provide results about optimal consumption $c^*$ vs. wealth $w$ for the following set of parameters: $\pi=1,$  $\gamma=3,$ $\mu=0.08,$ $r=0.02,$ $\rho=0.02,$ $\sigma=0.16$ for the problem where the part of the wealth invested into risky assets is fixed, $\theta=0.6.$ As a result, we expect to see convergence results with pension $\pi$ to results without pension for big enough wealth $w.$ Also we  see that the line which represents the  solution with asset allocation (dashed line) is higher then line which characterizes the solution with fixed $\theta$ (dash-dot line). Though  in fact these curves are so close it is hard to distinguish them. This confirms  that the optimal strategy is better then any other strategy. At the same time, if we consider smaller values of the parameter $\eta$ (\ref{f25_3} or \ref{f25_2}), we should expect slower convergence of the results with pension to those without pension. This can be understood via PDE (\ref{PDE_c_opt}). When the parameter $\eta$ is close to $1$ the term that includes it dominates over the term with pension $\pi.$ As a consequence, we see that results with pension converge to those without pension for big wealth and as the parameter $\eta$ becomes smaller the term with pension $\pi$ has more weight in the equation. In order to see convergence we would need to consider a significantly bigger scale. In terms of optimal consumption we can explain this the following way. When the parameter $\eta$ is close to $1$ (see Figure \ref{f25_1}), it means that the current value of consumption dominates. As the smoothing factor becomes smaller $(\eta<<1)$ averaging has more weight and, as a result, the final answer will be smoother and will react less to current changes in values.  Significantly slower convergence is anticipated for small value of the smoothing factor, for example $\eta=0.01$ (see Figure \ref{f25_3}). This makes these results less interesting to explore. 

Above, and in \S \ref{comparison2} we remarked on how close the curves were, corresponding to $\theta=0.6$ and to variable $\theta.$ To understand why, we appeal to Appendix \ref{section_comparison}.
Figure \ref{f49_1} shows that for $\sigma=0.16$ and $\eta=1,$ consumption is insensitive to asset allocation. For $\eta=0.1$ or $0.01$ it is more sensitive to $\theta,$ but only when $\theta$ varies a lot.
%----------------------------------------
\section{Retirement spending plan}
\label{wealth_sim}
%----------------------------------------
\subsection{Wealth depletion }
In order to answer the question of how much a client can consume based on his habit during retirement, we need to simulate some representive scenarios.  We assume that the agent follows the optimal strategy, and then we will use the results obtained in previous sections. 
Below we provide some illustration of our numerical results for three different values of the smoothing factor $\eta=10^{-2}, 10^{-1}$ and $1.0$ (see Figures (\ref{f9})-(\ref{f17})).  
For all cases, numerical results were obtained for the following set of parameters, $\pi=1,$ $\gamma=3,$ $\sigma=0.16,$ $\theta=0.6,$ $\mu=0.08,$ $r=0.02.$ By changing initial wealth $w_0=10,\; 30 $  and initial habit $\bar c_0=2,\; 10,\; 20$ we look at how the wealth and optimal consumption change over retirement. 
\begin{figure}[H]\centering
\begin{subfigure}[c]{0.4\textwidth}
\includegraphics[scale=0.27]{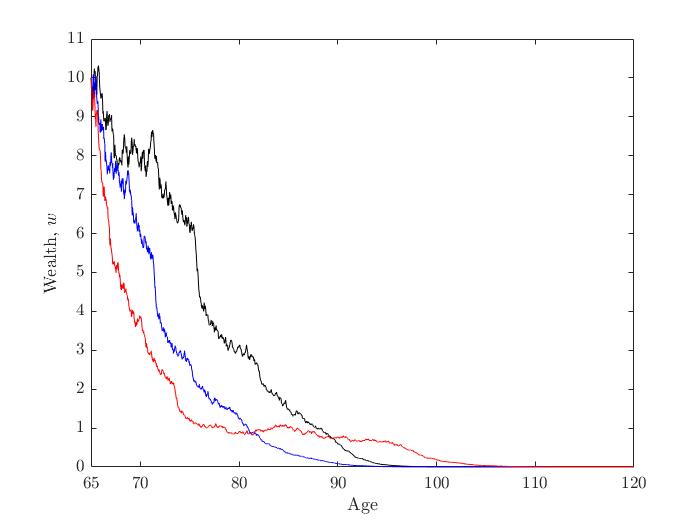}
\subcaption{}
\label{f9_1}
\end{subfigure}
\begin{subfigure}[c]{0.4\textwidth}
\includegraphics[scale=0.27]{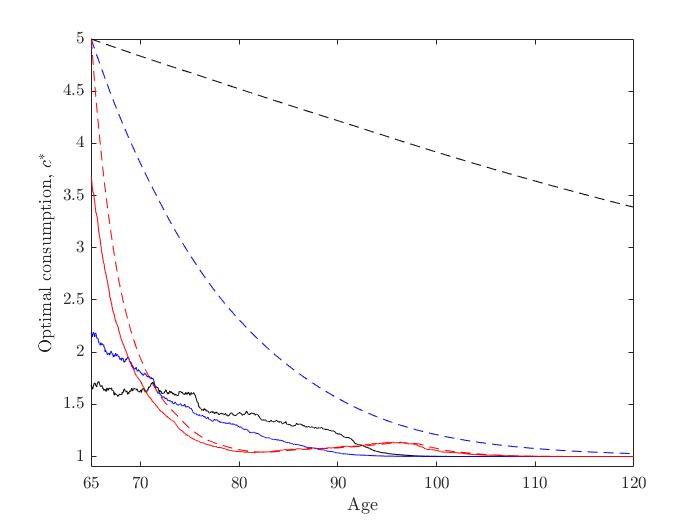}
\subcaption{}
\label{f9_2}
\end{subfigure}
\caption{Wealth and optimal consumption for initial value of wealth $w_0=10$ and initial EWA of consumption $\bar c_0=5$ for three parameters $\eta=0.01$ (black line), $\eta=0.1$ (blue line) $\eta=1.0$ (red line).}
\label{f9}
\end{figure}
In the picture on the left (see Figure (\ref{f9_1})) we can see how wealth changes over retirement between  $65$ and $120$yrs whereas on the right we can see the optimal consumption $c^*$ (solid lines) and habit $\bar c$ (dashed lines). For the small value of $\eta=0.01$ habit changes very slowly. It explains why the two black lines in the picture (see Figure (\ref{f9_2})) are a significant distance from each other. This set of pictures (Figure (\ref{f9})) shows the wealth $w$ and optimal consumption $c^*$ dynamics for initial wealth $w_0=10$ and initial habit $\bar c=5.$ As we can see the wealth's slope is relatively steep for all $\eta$ values and, for example for $\eta=0.01$ and $0.1,$ by the age $100$ the wealth becomes zero. For $\eta <<1$ consumption remains low during the whole time period and asymptotically approaches to the level of the pension by age $100.$ In other words, it means that the agent who follows the optimal strategy in this particular case, i.e. has relatively low initial wealth and  a low habit level, can't afford to spend more than twice the pension and has to remain approximately at this level during all the time period untill wealth ends around age $100$.
As the parameter $\eta$ increases the wealth slope becomes steeper (Figure (\ref{f9_1}))  because current values of consumption dominate over the averaging. This means the retiree consumes more early on.  At the same time, the curve which represents optimal consumption (see Figure (\ref{f9_2})), is closer to the EWA of the consumption line. This means that agent's optimal consumption better reflects habit behaviour. 

In addition, we can see that volatility is the highest when $\eta$ is the smallest. It means that, for example the black lines which represent $\eta=0.01$ have significantly bigger fluctuations then the red lines, which represent $\eta=1.0,$ on both pictures, Figure \ref{f9_1} and \ref{f9_2}.
 \begin{figure}[H]\centering
\begin{subfigure}[c]{0.4\textwidth}
\includegraphics[scale=0.27]{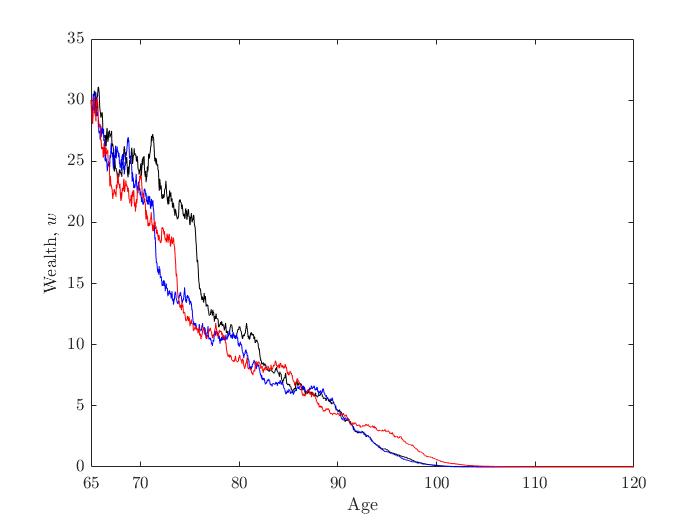}
\subcaption{}
\label{f19_1}
\end{subfigure}
\begin{subfigure}[c]{0.4\textwidth}
\includegraphics[scale=0.27]{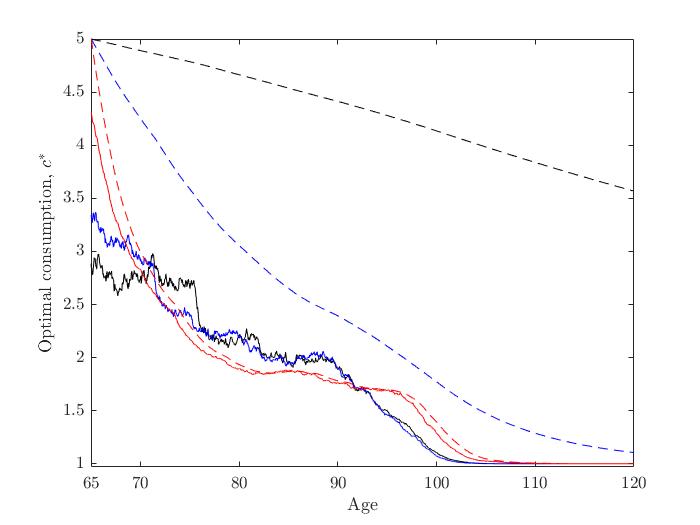}
\subcaption{}
\label{f19_2}
\end{subfigure}
\caption{Wealth and optimal consumption  for initial value of wealth $w_0=30$ and initial EWA of consumption $\bar c_0=5$ for three parameters $\eta=0.01$ (black line), $\eta=0.1$ (blue line) $\eta=1.0$ (red line).}
\label{f19}
\end{figure}
 Now, if we increase initial wealth up to $w_0=30$ and leave habit at the same level we can see that the wealth slope shows a more moderate decline over time. In this case we see that the optimal consumption can fluctuate at the constant level approximately $\sim 20$yrs with  the smoothing factor $\eta=1.0.$ 

The explanation of this behaviour is simple. It happens because we accepted a low level of habit $\bar c=5$ which means that the person did not consume a lot over previous time periods but has large enough wealth $w$, for his optimal consumption $c^*$ to show positive dynamics over the next time period.

Based on results for different parameters $\eta$ (Figure \ref{f9} or \ref{f19}) we can see that  changes in the agent's optimal consumption happen very fast for $\eta=1.0$ which means that at every time the agent can consume approximately the same as his habit. Moreover, we see that the level of consumption in this case is also higher than for smaller $\eta$ (solid red line Figure \ref{f9_2} or \ref{f19_2}). If we compare Figure \ref{f9_1} and \ref{f19_1} we conclude that the value of habit is very important. In cases where client does not consume a lot $(\bar c_0=5)$ we  see that higher initial wealth creates less difference between different models or parameters $\eta.$

We can consider another case where we vary the habit level $\bar c_0$ for a smoothing factor $\eta=1.0$ (see Figure \ref{f17}). The overall  observation is that with a higher initial habit $\bar c_0=20,$ less time is needed for wealth to reach zero and for consumption go down to the pension level. 

 \begin{figure}[H]\centering
\begin{subfigure}[c]{0.4\textwidth}
\includegraphics[scale=0.27]{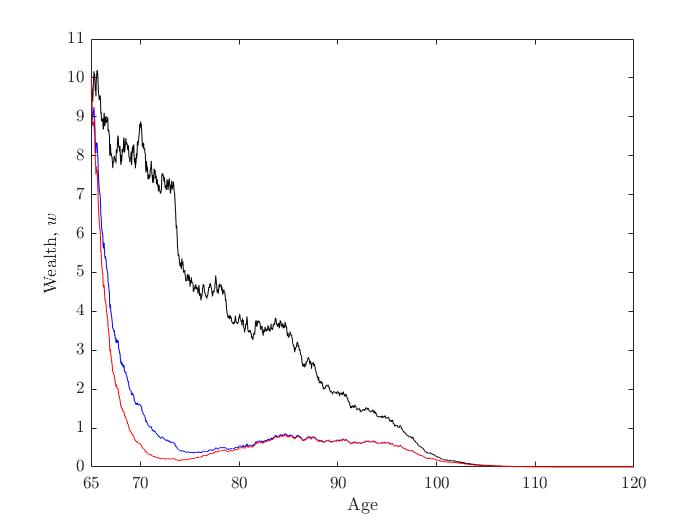}
\subcaption{}
\label{f17_1}
\end{subfigure}
\begin{subfigure}[c]{0.4\textwidth}
\includegraphics[scale=0.27]{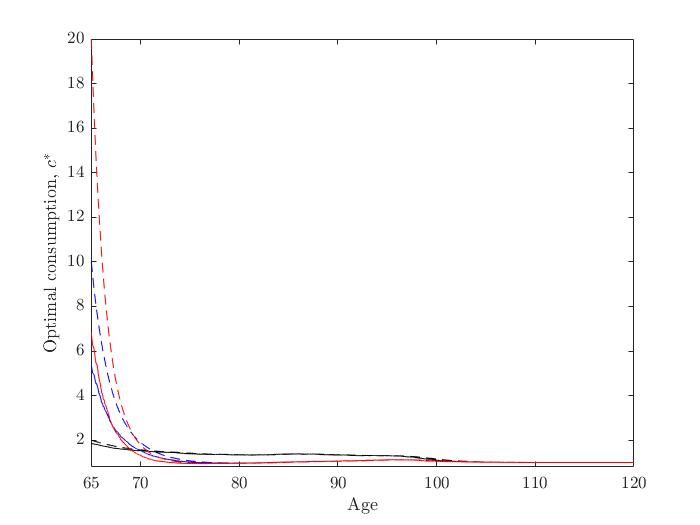}
\subcaption{}
\label{f17_2}
\end{subfigure}
\caption{Wealth and optimal consumption for an initial value of wealth $w_0=10,$ smoothing factor $\eta=1.0$ and three values of initial habit $\bar c_0=2$ (black line), $\bar c=10$ (blue line) $\bar c=20$ (red line).}
\label{f17}
\end{figure}
%

%----------------------------------------------
\subsection{Wealth depletion time (WDT)}
%----------------------------------------------
In the last section, we obtained simulation results using an Euler-Maruyama (EM) method. In this paragraph we will obtain WDT  $T_d$  using a PDE approach and will compare results with Euler's method. Following the same logic as in the article \cite{H} let us define  $T_d$ as expected wealth depletion time $ T_d(t,w,\bar c)=E[\tau|w_t=w,\bar c_t=\bar c,c_t^*=c^*]$ with an additional variable $\bar c$ that represents habit. Here $\tau$ is the first time that $w_t=0.$ We obtain a PDE
\beq
 T^d_t+ T^d_w( (\theta(\mu-r)+r)w+\pi-c^*)+ T^d_{\bar c}\eta (c^*-\bar c)+  \frac{1}{2}T^d_{ww}\theta^2\sigma^2w^2+1=0.
\eeq
Boundary conditions for this equation are similar to the previous case except  at the boundary with zero wealth, $w=0.$ WDT here is zero, $T^d(t,0,\bar c)=0.$ As in the previous case, in order to solve this equation, we use an implicit upwind method which is described in detail in Appendix \ref{appendix_FDS}. 

As a next step, we run more simulations for the different sets of initial wealth and habit and compute their mean and variance in order to see how the time depletion changes if we fix the wealth value and take the average over all habit values. We summarize our results in the Table \ref{table_WDT} where every entry is the average age at which wealth depletes. By ``age'' we mean 
%${\rm age}={\rm age\; of\; retirement}+{\rm time}$.  
${\textrm age}={\textrm age\; of\; retirement}+{\textrm time}$.  
As before, Table \ref{table_WDT} reflects an initial age of $65.$ As soon as  wealth reaches level zero we record our depletion time. In other words, this means that the retiree has spent all his savings and does not have income except his pension. 
%----------------------------------------------------------------------------------------------------
\begin{table}
  \centering
 \caption{Comparison between Monte Carlo and PDE results for $\bar c=10$.}
  \begin{tabular}{c|c|c|c|c|c|c}\hline 
    \multirow{2}{*}{\makecell{$w_0$\\  values}} & \multicolumn{3}{c|}{EM} &
                                                            \multicolumn{3}{c}{PDE}  \\ \cline{2-7}
            &  $\eta=0.01$              &     $\eta=0.1 $         &      $\eta=1.0$           &    $\eta=0.01$    &    $\eta=0.1 $   &  $\eta=1.0$                 \\ \hline
    $1$   &   $81.8\pm 2.1^*$  &  $70.6\pm 4.6$    &   $68.3 \pm 7.6 $       &       $81.5$        &    $70.2$        &  $68.7$            \\ \hline
    $5$   &   $91.5\pm 2.6$      &  $78.8\pm 5.1$    &   $76.5 \pm 13.1 $      &      $91.0$        &    $80.8$         &  $84.9$              \\ \hline
    $10$  &  $95.5\pm 2.7$      &  $84.0\pm 5.5$    &   $86.1\pm 14.4$        &      $95.3$        &    $87.8$         &  $93.8$               \\ \hline
    $20$  &  $99.4\pm 2.6$      & $90.3\pm5.5$      &   $99.3 \pm 4.0 $        &      $99.7$        &    $95.5$         &  $100.0$              \\ \hline  
    $35$  &  $102.6\pm 2.4$    & $95.5\pm 5.3$     &   $100.9 \pm 2.6 $      &     $102.8$       &   $100.7$         &  $103.4$                \\ \hline  
    $50$  & $104.5\pm 2.3$     & $99.0\pm 4.5$     &   $102.4 \pm 2.7 $      &     $104.6$       &   $103.4$         &  $104.9$              \\ \hline  
    $75$  &  $105.8\pm 2.1$    & $102.1\pm 3.9$   &   $103.7 \pm 2.7 $      &     $106.1$       &   $105.1$          & $105.7$              \\ \hline  
 \multicolumn{7}{c}{$*$ WDT range mean 
%$\pm {\rm std}$}
$\pm {\textrm std}$}
  \end{tabular}
  \label{table_WDT}
 \end{table}
%---------------------------------------------------------------------------------------------

To compare the EM and PDE methods, we choose three values of the smoothing factor $\eta=[0.01\; 0.1\; 1.0]$. This parameter reflects how fast habit reacts to changes in consumption. Based on the results from the Table \ref{table_WDT} we can see that the wealth depletion time function shows nonlinear behaviour, which could be subject for the additional analysis in later research. When we set this parameter equal to $1$ we see that averaging does not dominate. When this parameter decreases the averaging plays a greater role. In this paper we consider two more values of this parameter, namely $\eta=0.1$ and $0.01.$  The first column in the table represents the initial wealth $w_0$ values. In our case we examined seven different values $[1\; 5\; 10\;20\;35\;50\; 75]$. 
For instance, if we  look at the third row and second column, the corresponding values 
%$(w_0,\;{\rm age})$ 
$(w_0,\;{\textrm age})$ 
are $(5,\; 91.5\pm 2.6)$ which means that if the initial wealth equals $5$ then the average time when the wealth depletes is $91.49$ yrs. with a standard deviation $2.6.$  The trend for every fixed smoothing factor $\eta$ is, as initial wealth $w_0$ increases, depletion time also increases. Moreover, for small values of wealth for different values of the parameter $\eta$ the difference in WDT is very big. For $w_0=1$ it is $\sim 67$yrs for $\eta=1.0$ and $\sim 82$yrs for $\eta=0.01.$ 

At the same time, as initial wealth grows, the difference in depletion time becomes less significant. For example, there is no significant difference between depletion age for big wealth $75$ as $\eta$ varies with average values varying within the range $(102,\; 106)$ for EM and  $(105, 106)$ for PDE. If we compare the EM and PDE results we see that PDE solution is included in our EM intervals.  
%----------------------------------------------------------------------------------------------------
\begin{table}
  \centering
 \caption{Numerical  PDE results for $\eta=0$.}
  \begin{tabular}{c|c|c|c|c|c|c|c}\hline 
   $w_0$  &      $1$      &     $5$      &      $10$      &    $20$        &     $35$      &    $50$      &  $75$             \\ \hline
    PDE     &   $85.7$   &  $93.0$   &    $96.8$   &  $100.6$    &  $103.4$  &  $105.1$  & $106.4 $        \\ \hline
    \multicolumn{8}{c}{WDT range mean $\pm$ std$=100.5\pm 9.6.$ }\\ \hline
  \end{tabular}
  \label{table_WDT_eta0}
 \end{table}
%---------------------------------------------------------------------------------------------

There is one more case of smoothing parameter $\eta=0$ which means that we do not take habit into consideration. Table \ref{table_WDT_eta0} shows these results for the same set of initial wealth as in the previous case. The trend is the same as before, i.e. as soon as wealth increases WDT also increases. Comparing WDT for small values of wealth, with and without habit, it is greater in the second case because we do not count on habit and just follow the optimal strategy. 

We can explain this as follows. Assume that we have two clients and one of them follows the optimal strategy ignoring habit. According to the table \ref{table_WDT_eta0} if he has initial wealth $5$ at the moment of his retirement then he will spend all his money by age $93.$ Assume that another client follows the optimal strategy under habit formation model. In this case if he has the same initial amount of wealth $5$ he will spend all his money by age $91$ if the parameter value is $\eta=0.01$ or by age of $ 85$ for $\eta=1.0.$ As we know the greater the value of the smoothing parameter we have, the faster consumption will react to habit changes which means that the person with a higher value of the parameter $\eta$ can consume more. Therefore for big values of initial wealth the difference in WDT between $\eta$ is not that big.

Now, we provide some graphs (see Figures (\ref{WDT__}) - (\ref{WDT_3})) that illustrate how the results change for different smoothing factor values.
 \begin{figure}[H]\centering
\begin{subfigure}[c]{0.4\textwidth}
\includegraphics[scale=0.27]{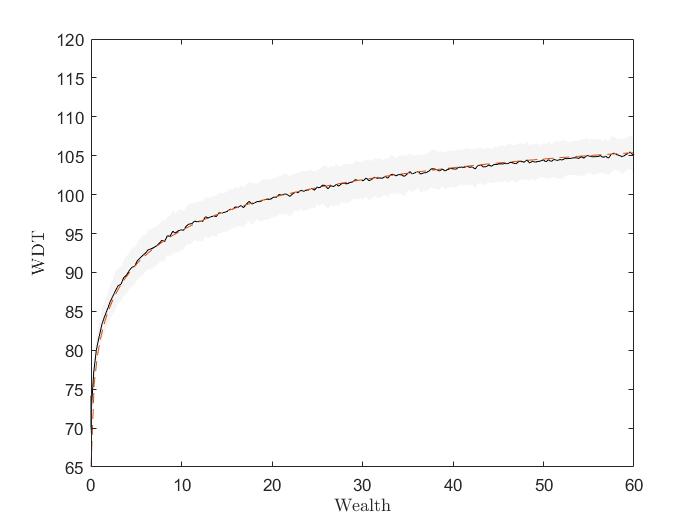} 
\caption{$\eta=0.01$ }
\label{WDT_1}\label{f20_1}
\end{subfigure}
\begin{subfigure}[c]{0.4\textwidth}
\includegraphics[scale=0.27]{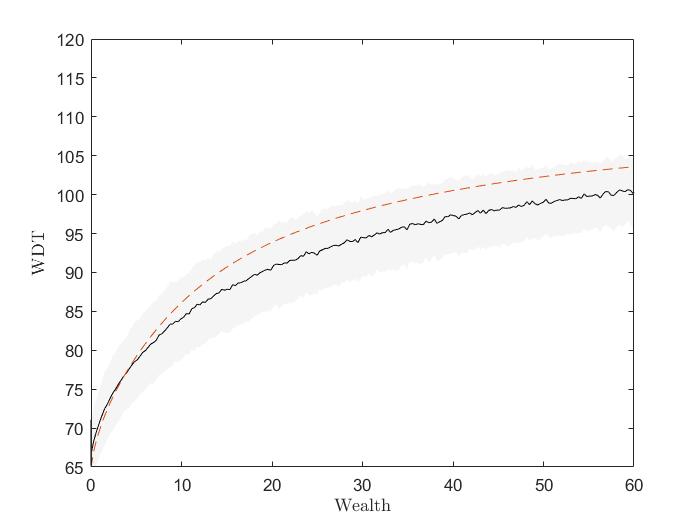} 
\caption{$\eta=0.1$}
\label{WDT_2}
\end{subfigure}
\begin{subfigure}[c]{1.0\textwidth}\centering
\includegraphics[scale=0.27]{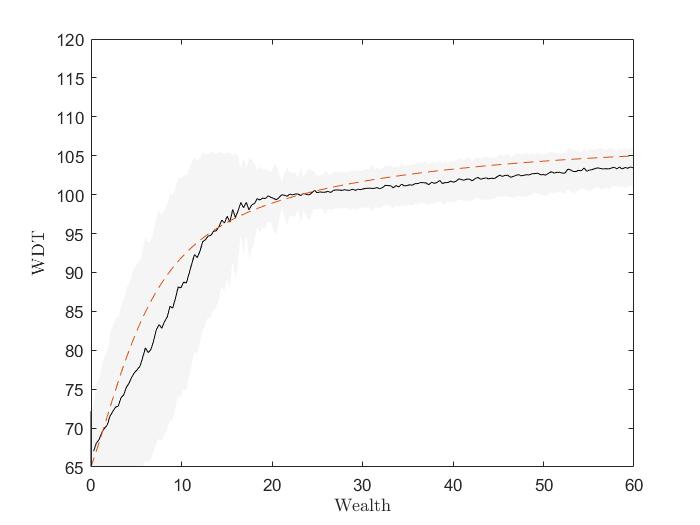}
\caption{$\eta=1.0$ }
\label{WDT_3}
\end{subfigure}
\caption{WDT vs. initial wealth for fixed habit $\bar c=10.$}
\label{WDT__}
\end{figure}
For all three pictures, the solid line represents the EM solution, the dashed line the PDE solution and the shadowed area  represents the standard deviation of the stochastic solution over all values of initial habit $\bar c_0$.
For comparison we fixed habit at the level $\bar c=10.$ As we can see the case with $\eta=0.01$ (Figure \ref{WDT_1}) has the smallest shadow area and the two lines are almost indistinguishable which means  both solutions, EM and PDE, are in the good agreement. As the smoothing factor increases to $\eta=0.1,$ the shadowed area becomes bigger as well (Figure \ref{WDT_2}) and the lines are not as close as in the previous case. But still the PDE solution lies within the stochastic solution boundaries. The shadowed area reflects the stability of the particular solution, which is why for $\eta=1.0$ (Figure \ref{WDT_3}) we see an irregular shadow.
%----------------------------------------
\section{Liquid wealth vs. annuity}
\label{wealth}
%----------------------------------------
Suppose that there is a client who at the time of retirement, for example at age $65,$ has a certain amount of wealth $w$ (non annuitized) and pension $\pi,$ such as Social Security benefits. There are several papers where authors showed that under certain conditions it is better to convert their wealth into annuities (see for example \cite{Rei}). In this paper we answer the question, should the endowment be annuitized at the moment of retirement or not, under HFM. First, let us define the annuity equivalent wealth (AEW). Based on the definition provided in the article \cite{MH}, the AEW is the quantity $\hat w$ that satisfies the following equation:
\be
V(0,\hat w,\bar c, \pi)= V(0,0,\bar c,\pi+w/a_x)
\label{AEW}
\ee
where $V$ is the maximized discounted lifetime utility function, $w$ represents initial liquid wealth, $a_x$ is the annuity factor and $x$ is the client's current age. The idea is that a client with initial wealth $w$ and a certain utility should retain the same utility level if he decides to annuitize the entire wealth.  In this paper instead we explore  whether it is reasonable to annuitize part of the retiree's wealth into an annuity under HFM.  We compare  two value functions, namely $V(0,w,\bar c,\pi)$ and $V(0,w-\Delta w,\bar c,\pi+\Delta w/a_x)$. Before we move forward let us say a couple words about the annuity factor $a_x.$
We can define $a_x$ as follows
\be
a(r,x,m,b)=\disp\frac{b \Gamma(-rb, exp(\frac{x-m}{b}))}{exp(r(m-x)-exp(\frac{x-m}{b}))}
\ee
where all values reflect the Gompertz formulation (see for example \cite{M}) and all parameters are as described in \S \ref{theor_backgr}. We can think about the following financial definition of the annuity factor $a_x.$ In case the retiree converts  some or all of his wealth into annuities, the annuity factor is their unit price $a_x.$ 

Now, let us introduce new notation. We define  the difference $\Delta V$ on the graphs as follows 
\be
 \Delta V=V(0,w,\bar c,\pi) - V(0,w-\Delta w,\bar c,\pi+\Delta w/a_x).
 \label{difference}
 \ee
Below we provide some numerical results for our model where we take into consideration the client's habit $\bar c$. For the case with smoothing parameter $\eta=0,$ which means that
we do not include habit, calculations were done in the paper \cite{MH} but for the case where the client does not invest anything into risky assets, i.e. $\theta=0.$ In this paper we will check if the inequality, that those authors derived in their article, holds for HFM
\be
 V(0,w,\bar c,\pi) < V(0,w-\Delta w,\bar c,\pi+\Delta w/a_x). 
 \label{inequality}
\ee
In other words, this inequality means that for the client at age $65$ it is better to annuitize part of the wealth $\Delta w$  when  $\eta=0.$ At the same time, as the parameter $\eta$ increases the picture changes.  As soon as the retiree invests more into risky assets, for example $20\%$ or $60\%$ (see Figure \ref{f4_01}) of his portfolio, the difference becomes positive for wealth below a certain level. On the graphs we see the relationship between the difference $\Delta V$ (\ref{difference}) and wealth $w$. Results were obtained for fixed asset allocation and habit value $\bar c\approx 10$ and for four values of the parameter $\eta=[0 \;0.01\; 0.1\; 1.0]$ (see Figure \ref{f4_01}). 

As we can see, for example from the Figure (\ref{f4_01a}) where only $20\%$ of portfolio is invested into risky assets, and $\theta=0.2$, when the effect of current consumption is strong (eg. $\eta\sim 1$) little benefit accrues from the converting wealth into annuities at a retirement age of $65$yrs. For $\eta=0.01$ the impact of annuitization is significant at large wealth. 
 \begin{figure}[H]\centering
\begin{subfigure}[c]{0.4\textwidth}
\includegraphics[scale=0.27]{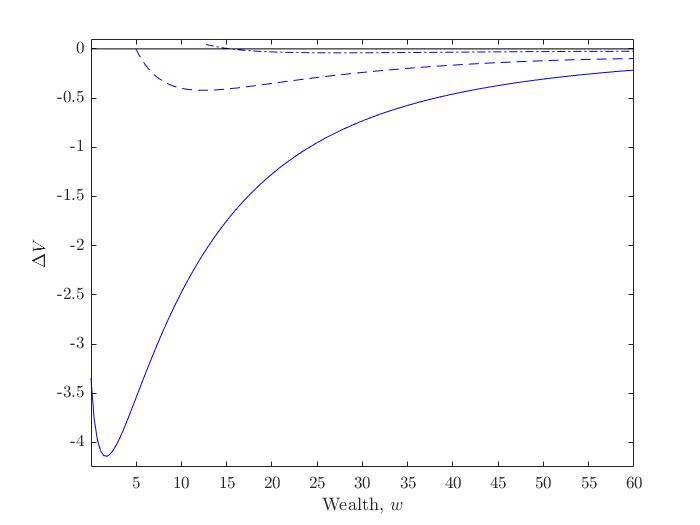} 
\subcaption{$ \theta=0.2$}
\label{f4_01a}
\end{subfigure}
\begin{subfigure}[c]{0.4\textwidth}
\includegraphics[scale=0.27]{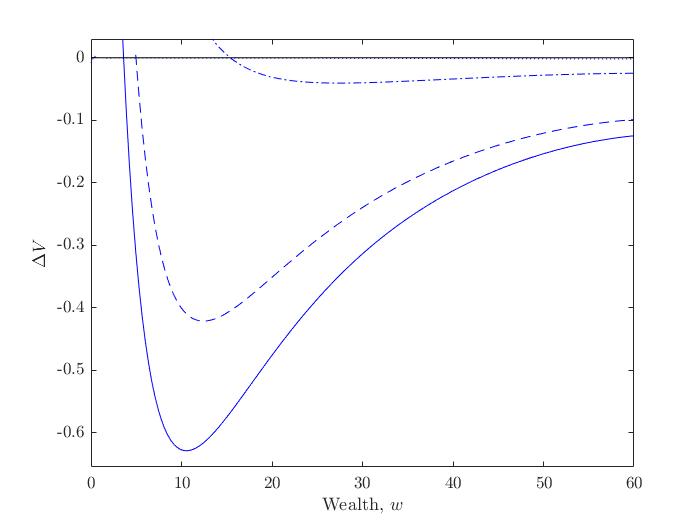}
\subcaption{$\theta=0.6$}
\label{f4_02}
\end{subfigure}
\caption{Difference $\Delta V$ vs. wealth $w$ for fixed habit $\bar c\approx 10,$ smoothing parameter  $\eta=0$ (the blue solid line), $0.01$ (the blue dashed line), $0.1$ (the blue dash-dot line) and $1.0$ (the blue dotted line).  The black horizontal line represents zero level in $\Delta V$.}
\label{f4_01}
\end{figure}
Similar behaviour can be  seen for the parameter value  $\theta=0.6$ (see Figure \ref{f4_02}).  The observation that we can make after taking into consideration all numerical results is the following: as soon as the retiree decides to invest more into risky assets, for example $\theta=0.6,$ the behaviour of the function that represents the difference $\Delta V$ changes and becomes positive below a certain level, which means  that client should not convert his endowment into annuities unless he has at least a minimum wealth level. 
%----------------------------------------------------------------------------------------------------
\begin{table}
  \centering
 \caption{Wealth values that show when to annuitize.}
  \begin{tabular}{c|c|c|c|c|c|c}\hline 
    \multirow{2}{*}{$\theta$  values} & \multicolumn{2}{c|}{ $\bar c=1$} &
                                                            \multicolumn{2}{c|}{$\bar c=5$}  & 
                                                                   \multicolumn{2}{c}{$\bar c=10$}  \\ \cline{2-7}
           &$\eta=0.01$ &$\eta=1.0 $  &   $\eta=0.01$    & $\eta=1.0$  &    $\eta=0.01$    & $\eta=1.0$     \\ \hline
    0     &    $\star$  &     $\star$           &   $\star$           & $7.77$        &       $\star$     & $14.47$      \\ \hline
    0.2  &    $\star$  &     $\star$           &   $\star$           &  $8.47$   &  $\star$ &  $15.53$       \\ \hline
    0.6  &    $3.18$       &     $2.82$           &   $4.94$        & $12.35$   & $5.29$ & $21.17$  \\ \hline  
      \multicolumn{7}{c}{$\star$ means that the difference $\Delta V$ is negative for all $w$}  \\ \hline     
  \end{tabular}
  \label{table}
 \end{table}
%---------------------------------------------------------------------------------------------

To summarize the results,  Table \ref{table} shows when it is better to annuitize wealth for two values of the smoothing factor $\eta=0.01$ and $1.0,$ some values of asset allocation $\theta=0,\;0.2,\; 0.6$ and three values of habit $\bar c\approx 1,\; 5,\; 10$ (see Table \ref{table}). For instance, if we consider a small value of habit $\bar c=1.0$ it is definitely better to annuitize wealth for almost all values of asset allocation whereas for $\bar c=10$ this is not true. 

 For example, value $15$ in the fifth row and fifth column means that for the smoothing factor $\eta=1,$  with habit $\bar c\approx 5$ and 60\% of wealth invested into risky assets, the minimum value of the wealth when it is reasonable to annuitize at age $65$ is approximately $15,$ which is significantly greater than pension income.
The last column in the table shows that if the habit $\bar c$ is relatively high, in our case $\bar c=10,$ then it does not matter how much you invest into risky assets, the response is very fast, i.e. we have less impact from averaging ($\eta\approx 1$), and for all cases the value of wealth when it is reasonable to convert part of the wealth into annuities is greater than $15.$  In other words, the wealth should be $15$ times greater than pension income, such as $17$ for $\theta=0.2$ or $w\approx 20$ for $\theta=0.6.$ 

The last set of pictures (see Figure \ref{annuity_4}) represents the scaled version of our numerical results, where scaling was described in one of the previous sections \ref{subsection_1cs}. We use two values of the smoothing parameter $\eta=0.01$ and $1.0.$ The y-axis represents the difference $\Delta V$ (\ref{difference}) and the x-axis is a logarithm of the ratio between wealth $w$ and habit $\bar c$, $\log(w/\bar c).$  We see  two areas which show that the question about wealth annuitization does not have an unambiguous answer. Every line represents a different value of habit $\bar c$ that increases in the direction pointed to by the arrow. In the left Figure (\ref{annuity_4a}) we see results for $\eta=0.01$ which means that averaging dominates and habit adapts slowly. Here there are regions where annuitization is favourable. Whereas in the right picture (\ref{annuity_4b}) averaging has the least impact on our solution and, as a consequence, habit adapts repidly. Here the answer is that for small wealth it is not optimal to annuitize part of the wealth. 

Finally, we can say that all results show that the presence of habit makes annuitization less reasonable for small values of wealth $w$ while it becomes a good choice for large wealth. At the same time, the more a retiree invests into risky assets the less favourable it will be to annuitize part of the wealth immediately at the age of $65.$ 
 \begin{figure}[H]\centering
\begin{subfigure}[c]{0.4\textwidth}
\includegraphics[scale=0.27]{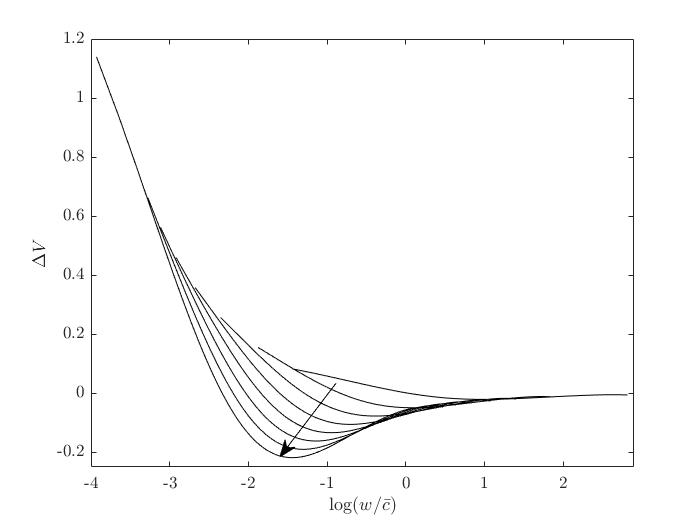}
\subcaption{$\eta=0.01$}
\label{annuity_4a}
\end{subfigure}
\begin{subfigure}[c]{0.4\textwidth}
\includegraphics[scale=0.27]{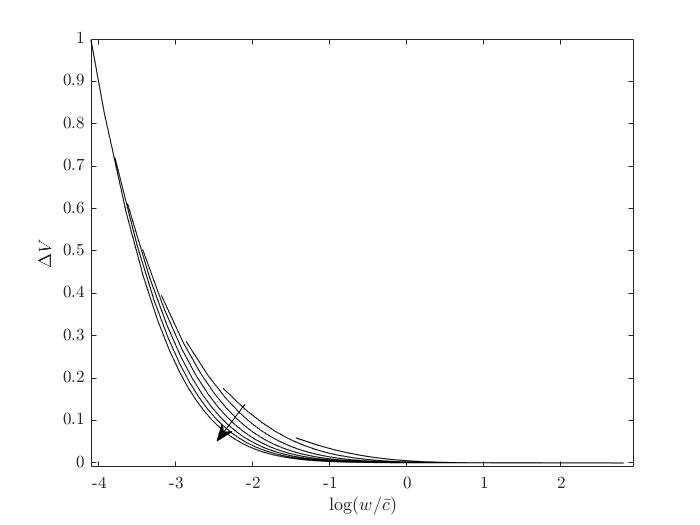}
\subcaption{ $\eta=1.$}
\label{annuity_4b}
\end{subfigure}
\caption{Difference $\Delta V$ vs. scaled wealth $w$ for fixed asset allocation $\theta=0.6.$ The arrow shows the direction of increasing habit $\bar c.$}
\label{annuity_4}
\end{figure}
%
%------------------------------------------------
\section{Conclusion}
%------------------------------------------------
In this paper we solved the retirement optimization problem under a habit formation model. The goal was to explore how the exponentially weighted average of consumption affects the optimal consumption strategy by changing different parameters values, such as the smoothing factor $\eta,$ asset allocation $\theta$ and volatility $\sigma$. Our results show that there is a strict correspondence between these parameters. Also we answered  the question of how much the potential client can spend during his retirement if he follows the optimal strategy based on his habit $\bar c$. Finally, we provided some analysis about how this optimal spending strategy changes, by varying the initial values of wealth $w_0$ and habit $\bar c_0$ and analyzing under what conditions the client should convert his wealth at the age of retirement into annuities.

%----------   The  Bibliography    ------------

%----------------------------------------------------------------------------
\newpage
\begin{appendix}
   %-----------------------------------------------
    \section{Theoretical Background}
    \label{appendix_theorem}
   %-----------------------------------------------
In order to find solution of our OP we need to understand if there exist a function which satisfies our constraints and if it is unique. In other words, there are some questions which should  be answered before solving the problem:
\begin{enumerate}
\item Prove that  a smooth solution of an HJB is a either supermartingale or martingale.
\item Prove that supermartingale or martingale solution solves the OP.
\end{enumerate}
Let us start from proving first statement.  Suppose we have a smooth function $\tilde V$ that solves the HJB equation. Define a stochastic process $Y_t$ by
\be
Y_t\equiv \tilde V(t,w_t,\bar c_t)+\int_0^t e^{-\rho t^{\prime}}\tensor[_{t^{\prime}}]{p}{_x} u\left(\frac{c_{t^{\prime}}}{\bar c_{t^{\prime}}}\right)\d t^{\prime}.
\label{TB1}
\ee
The goal is to prove that the stochastic process $Y_t$ is a supermartingale for any admissible choice of consumption $c_t$ and martingale for the optimal  value of consumption $c^*.$ There are at least two approaches how this statement can be proved, for example, we can prove  it by taking expectation $E[Y_t|\EuScript{F}_s]$ and then using Fatou's lemma \cite{D}. Here, we will provide another approach.
Take the differential of this expression (\ref{TB1}) and compute in the drift and volatility terms. Then, plug the expression for wealth and consumption rate dynamics (\ref{TB0}).
After performing some calculations we get the following
\beq
\begin{array}{r}
\vspace{0.2cm}
\d Y_t=\Bigl\{\tilde V_t+\tilde V_w(\theta(\mu-r)+r)w_t+ \tilde V_w(\pi-c_t)+\tilde V_{\bar c}\eta (c_t-\bar c_t)+ \Bigr.\\ 
\Bigl.\frac{1}{2}\tilde V_{ww}\theta^2\sigma^2w^2+e^{-\rho t}\tensor[_t]{p}{_x} u\left(\disp\frac{c_t}{\bar c_t}\right)\Bigr\}\d t  +\tilde V_w\theta \sigma w_t \d W_t. 
\end{array}
\label{TB3}
\eeq
Then integrate (\ref{TB3}) $\int_0^t$ assuming that function that represents drift term can be written as $f(t,w_t,\bar c_t)$ and volatility term as $g(t,w_t,\bar c_t).$
\be
\int\limits_0^t\d Y_{t^{\prime}}=\int\limits_0^tf(t^{\prime},w_{t^{\prime}},\bar c_{t^{\prime}})\d t^{\prime}+\int\limits_0^t g(t^{\prime},w_{t^{\prime}},\bar c_{t^{\prime}})\d W_{t^{\prime}}
\label{TB4}
\ee
Using the additivity property and taking expectation $E[\ldots|\EuScript{F}_s] \quad \forall s\le t$ after that,  we obtain the following expression
\beq
 E[Y_t|\EuScript{F}_s]= E[Y_s|\EuScript{F}_s]+E\Bigl[\int\limits_s^t f(t^{\prime},w_{t^{\prime}},\bar c_{t^{\prime}})\d t^{\prime}|\EuScript{F}_s\Bigr]+\nonumber \\
 E\Bigl[\int\limits_s^t g(t^{\prime},w_{t^{\prime}},\bar c_{t^{\prime}})\d W_{t^{\prime}}|\EuScript{F}_s\Bigr].
 \eeq
 Now, recall that the drift term $f(t,w,\bar c)$ is zero in the case of a martingale  or negative if we have a supermartingale, then if we assume that integrand $g(t,w,\bar c)$ is square integrable then 
\be
E\Bigl[\int\limits_s^t g(t^{\prime},w_{t^{\prime}},\bar c_{t^{\prime}})\d W_{t^{\prime}}|\EuScript{F}_s\Bigr]=0.
\ee
Then $ E[Y_t|\EuScript{F}_s]\le Y_s.$
The statement is proved.  Now, let us reformulate the third statement. 

\begin{theorem} (verification theorem): 
\label{verification_theorem}
Suppose  $\exists \tilde V:[t,T]\times \RR^+ \times \RR^+\to \RR$ which is $C^{1,2}.$
The objective has the form (\ref{TB5}).
Suppose that $\forall c_t \in  \mathscr{A}(w)$ where $ \mathscr{A}(w)\equiv \{c : c\,$ is  admissible  with regards to $ w\}$ the stochastic process 
\beq
Y_t\equiv \tilde V(t,w_t,\bar c_t)+\int\limits_{0}^t e^{-\rho s}\tensor[_s]{p}{_x}u\left(\frac{c_s}{\bar c_s}\right)\dd s \quad {is \; a\; supermartingale}
\label{verification}
\eeq
and $\tilde V(T,w,\bar c)=0,$ also $\exists c^* \in  \mathscr{A}(w)$ such that the process $Y_t$ is a martingale. Then $c^*$ is optimal, and  solution of the problem is 
\beq
\tilde V(t,w,\bar c)=\sup_{c_t} E\left[\int_t^T e^{-\rho s}\tensor[_s]{p}{_x}u\left(\frac{c_s}{\bar c_s}\right)\dd s|w_t=w,\bar c_t=\bar c\right].
\label{verification_}
\eeq
The consumption $(c_t)_{t>0}$ is admissible for wealth $w$  if the wealth process $w_t$ remains positive at all time. 
\end{theorem}
\begin{proof}
The supermartingale property states that $E[X_t|\EuScript{F}_s]\le X_s,$ $0\le s\le t,$ $E[X_s]<\infty$ (see, for example \cite{D}). The goal is to prove that the following inequality holds
\be
E[Y_T|\EuScript{F}_t]\le Y_t.
\ee
By plugging values ``T'' and ``t'' consequently into equation (\ref{verification}) and performing simple calculations we can get 
\beq
\tilde V(t,w,\bar c)\ge E\Bigl[\int\limits_t^T e^{-\rho t^{\prime}}\tensor[_{t^{\prime}}]{p}{_x} u\left(\frac{c_{t^{\prime}}}{\bar c_{t^{\prime}}}\right)\d t^{\prime}|\EuScript{F}_t\Bigr].
\eeq
The equality comes from the definition of optimal value. The theorem  is proved.
\end{proof}

%-----------------------------------------------
 \section{Finite Difference Scheme}
 \label{appendix_FDS}
%-----------------------------------------------     
	In this part we explain what kind of numerical method we use and derive the corresponding formulas. In order to solve PDEs  we will choose the implicit upwind method. The idea of this method is to use  a forward difference for the time derivative and forward or backward diffference for other variables depending on the sign of the coefficient in front of every first order derivative. For this scheme we will implement a generalized upwind method where we add and subtruct the absolute value of the coefficient that changes sign.  Let us set up the grid as follows $V(t,w,\bar c)=V(t_n,w_j,\bar c_k)$ and indices change $j=1\ldots M,$ $n=1\ldots N$ and $k =1 \ldots K.$ We introduce a new time variable as follow: $t_n=T-n\Delta t$ where $T$ is the terminal time.
	
	The approximation scheme for the equation (\ref{PDE_c_opt}) will be  the following
	\beq
	 \hspace{-2cm}\frac{V^{n+1}_{j,k}-V^{n}_{j,k}}{\Delta t}+(\rho +\lambda_{t_{n+1}+x})V^{n+1}_{j,k}+\nonumber \\
	 \frac{\alpha^{n}_{j,k}+|\alpha^{n}_{j,k}|}{2}\frac{V^{n+1}_{j,k}-V^{n+1}_{j-1,k}}{\Delta w}+\frac{\alpha^{n}_{j,k}-|\alpha^{n}_{j,k}|}{2}\frac{V^{n+1}_{j+1,k}-V^{n+1}_{j,k}}{\Delta w}+\nonumber \\ 
	\frac{ \beta^{n}_{j,k}+| \beta^{n}_{j,k}|}{2}\frac{V^{n}_{j,k}-V^{n}_{j,k-1}}{\Delta \bar c}+\frac{ \beta^{n}_{j,k}-| \beta^{n}_{j,k}|}{2}
	 \frac{V^{n}_{j,k+1}-V^{n}_{j,k}}{\Delta \bar c}- \nonumber \\
	  \frac{1}{2}\theta^2\sigma^2w_j^2 \frac{V^{n+1}_{j+1,k}-2V^{n+1}_{j,k}+V^{n+1}_{j-1,k}}{\Delta w^2}- u\left (\frac{\zeta^{n}_{j,k}}{\bar c_k}\right)=0
	\label{pde_scheme}
	\eeq
	where
	\beq 
	\alpha^{n}_{j,k}=-\{(\theta(\mu-r)+r)w_j+\pi-\zeta^{n}_{j,k}\}  \nonumber \\ 
	\beta^{n}_{j,k}=-\eta(c^*-\bar c_k)=-\eta(\zeta^{n}_{j,k}-\bar c_k) \nonumber
	\eeq
	\be
	\zeta^{n}_{j,k}=c^*=\bar c_k^{\frac{\gamma-1}{\gamma}}\left(\frac{V^{n}_{j,k}-V^{n}_{j-1,k}}{\Delta w}-\eta \frac{V^{n}_{j,k+1}-V^{n}_{j,k}}{\Delta \bar c}\right)^{-\frac{1}{\gamma}}.
	\label{beta}
	\ee
	The rest of the parameters are constants and  we will take the following values for them $\pi=1,$ $\rho=0.02,$ $\eta=[10^{-2}\; 10^{-1}\; 10^0],$ $\theta=[0.2\; 0.6 \; 0.9],$ $\sigma=[0.16 \; 0.50 \; 0.75],$ $\mu=0.08,$  $r=0.02$ and $\gamma=3$ .  For every case study and test we specify which values we used.
	
	\begin{itemize}
	%----------------------------------------
	\item Boundary condition at $w=w_1=0.$ \\
	%----------------------------------------
	The equation (\ref{BC_1}) can be approximated as follows
	\beq
	\frac{V^{n+1}_{1,k}-V^{n}_{1,k}}{\Delta t}+(\rho +\lambda_{t_{n+1}+x})V^{n+1}_{1,k}+\alpha^{n}_{1,k}\frac{V^{n+1}_{2,k}-V^{n+1}_{1,k}}{\Delta w}+\nonumber \\ 
	 \frac{ \beta^{n}_{1,k}+| \beta^{n}_{1,k}|}{2}\frac{V^{n}_{1,k}-V^{n}_{1,k-1}}{\Delta \bar c}+
	\frac{ \beta^{n}_{1,k}-| \beta^{n}_{1,k}|}{2} 
	 \frac{V^{n}_{1,k+1}-V^{n}_{1,k}}{\Delta \bar c}- u\left (\frac{c^*}{\bar c_k}\right)=0  \nonumber
	\eeq
	where $\alpha^{n}_{j,k}=-\pi+\zeta^{n}_{j,k}$  and $\beta^{n}_{1,k}=-\eta(\zeta^{n}_{1,k}-\bar c_k),$ $\zeta^{n}_{1,k}=\min(\pi,\bar c_k).$ 
	%----------------------------------------------
	\item Boundary condition at $w=w_{\max}=w_M.$ \\
	%----------------------------------------------
	\be
	\frac{V^{n+1}_{M,k}-V^{n+1}_{M-1,k}}{\Delta w}=0. \nonumber
	\ee
	\end{itemize}
%---------------------------------------------------------------
\section{Error analysis}
\label{error}
%---------------------------------------------------------------
In this paragraph we are going to provide some intuition about the numerical error, the so-called discretization error. Since, in order to solve our optimization problem,  we use the approximation scheme which is described in detail in Appendix \ref{appendix_FDS}, it is reasonable to check how big the errrors are and how changes in parameters can affect on them. So, we calculate an $L_2-$ norm (see Tables \ref{t1} and \ref{t4}). 
In our case, we have three variables and, as a consequence, the disretization error can be estimated by the following inequality:
\be
e(\Delta t,\Delta w, \Delta \bar c)\le c_1 (\Delta t)^{p_1}+c_2 (\Delta w)^{p_2}+c_3 (\Delta \bar c)^{p_3}
\label{I}
\ee
where $\Delta t,$ $\Delta w, $ $\Delta \bar c$ are step sizes over variables time $t,$ wealth $w$ and habit $\bar c$ respectively, $c_1,\, c_2$ and $c_3$ are finite constants. The time step can be calculated as follows
\be
\Delta t =\frac{b-a}{N-1}
\label{h}
\ee
where time belongs to the interval $t\in[a\ldots b]$ and $N$ is the number of nodes. The other steps $\Delta w$ (number of nodes $M$) and $\Delta \bar c$  (number of nodes $K$) can be computed by analogy. Parameters $p_1,$ $p_2$ and $p_3$  represent order of accuracy (OA). In the general case OA quantifies the rate of convergence of a numerical approximation of a differential equation to the exact solution. In our case we will compare two numerical solutions, with single and doubled step size. 
It can be said that numerical  scheme is accurate of order $(p_1,p_2,p_3),$ which means that a scheme is accurate of order $p_1$ in time, order $p_2$ in wealth and order $p_3$ in habit (see \cite{S}).
Let's estimate the error ratio ($ER$), that is simply ratio of two consequtive norms, for our scheme. First, in the formula (\ref{h}) we omit $1$ in denominator because for $N$ big enough it's not important. Then we choose parameters $p_1=p_2=p_3=p=1$ since we use corresponding approximations for a numerical scheme. Another simplification that we can accept is the following, we omit one term from the inequality (\ref{I}) if we choose a fine enough grid, for example over the time $t$, namely
\be
%\Delta t\to 0 {\rm \quad then \quad} c_1\Delta t \ll  c_2 \Delta w+c_3 \Delta \bar c
\Delta t\to 0 {\textrm \quad then \quad} c_1\Delta t \ll  c_2 \Delta w+c_3 \Delta \bar c
\ee
Then, if we assume the step, for example over the habit, is propotional to the wealth step $\Delta\bar c=h\Delta w$ where $h$ is a constant then, by perfoming simple calculations we can derive the ER of our scheme. Create system of 2 equations and divide one by another, this operation will allow to get rid of the constant which will be the combination of another constants, namely $c^{\prime}_2+hc^{\prime}_3$ and the following equation for ER  can be obtained:
\be
\frac{e(M)}{e(2M)}=\left(\frac{\Delta w_{K}}{\Delta w_{2K}}\right)^p=\left(\frac{2M}{M}\right)^p=\left(\frac{4M}{2M}\right)^p=\ldots=2^p.
\ee

Since we do not have the exact solution for this problem we calculated the $L_2$ matrix norm increasing the number of nodes in a particular direction. The error ratio (ER) can be computed based on the following formulas:
\beq
\begin{array}{c}
L_2(n)=norm(c^{*}(t_0,2^{5+n},k)-c^{*}(t_0,2^{6+n},k)), \quad n=1,2,3,4  \vspace{0.3cm} \\
%\disp {\rm ER}=\frac{L_2(n)}{L_2(n+1)}
\disp {\textrm ER}=\frac{L_2(n)}{L_2(n+1)}
\end{array}
\label{Error}
\eeq
where $t_0$ is a fixed time moment and $k$ is the number of nodes over EWA of consumption. In other words we fix number of nodes over one variable, for example $\bar c,$ and
increase the number of nodes over another one. In our case it is wealth $w,$ so fix the number of nodes over the wealth for which we got the best results and change the number of nodes over the habit $\bar c.$ Results are provided for three parameters $\eta=10^{-2}, \, 10^{-1}$ and $10^{0}.$
\begin{enumerate}
\item Results for the first case are summarized in the Table (\ref{t1})). The number of nodes over the habit $\bar c$ was chosen to be $K=80.$
Every table shows results for all three values of the smoothing factor $\eta,$ $L_2$ represents values of $L_2$-norm, $ER$ is the error ratio which  was computed using formula (\ref{Error}). Numbers, for example $256/512,$ represent the number of nodes over the wealth $w$ for norm calculation. As follows from the explanation $2^p=2$ and since OA for our scheme  equals $1$ then we should expect that the $ER$ approaches $2.$ As we can see from the tables values  $ER$ for $L_2$ norm vary from $1.76$ for $\eta=1$ (see Table (\ref{t1}), column 7) up to $1.94$ for $\eta=10^{-1}$ (see Table (\ref{t1}), column 5).  
 As we can see from the tables, the results show stability and the value of $ER$ approaches $2.0.$    
%----------------------------------------------------------------------------------------------------
\begin{table}
  \centering
  \caption{$L_2$ norm. }
  \begin{tabular}{c|c|c|c|c|c|c}\hline
  \multicolumn{7}{c}{Error analysis for fixed \# of nodes over $\bar c.$ }  \\ \hline  
    \multirow{2}{*}{\#  of nodes} & \multicolumn{2}{c|}{ $\eta=0.01$} &
                                                            \multicolumn{2}{c|}{$\eta=0.1$} &
                                                                 \multicolumn{2}{c}{$\eta=1.0$} \\\cline{2-7}
                       & $L_2$ & $ER$ & $L_2$  & $ER$  & $L_2$  & $ER$  \\ \hline
    64/128    & $1.5$                      &  -    &   $3.04$                    &    -       & $1.3\cdot 10^{1}$  &      -    \\ \hline
    128/256  & $8.1\cdot 10^{-1}$   &  $1.84 $    &    $1.63$                    &   1.86   &$ 7.6$       & $1.76$ \\ \hline
    256 /512 & $4.3\cdot 10^{-1}$   &  $1.89  $   &    $8.58\cdot 10^{-1}$ &   1.91   & $4.2 $      & $1.82$  \\ \hline
    512/1024 & $2.2\cdot 10^{-1}$  &   $1.93$    &    $4.43\cdot 10^{-1}$ &   1.94   & $2.2 $      & $1.86$  \\ \hline
  \end{tabular}
  \label{t1}
 \end{table}
%---------------------------------------------------------------------------------------------
In order to achieve better results we need to do the same procedure in the other direction, over the habit $\bar c$ keeping the fine grid over time.
\item 
So, we will fix the number of nodes over the wealth for which results were the best, i.e. $M=1024$ and increase the number of nodes in another direction. Results are summarized  in Table \ref{t4} for three parameters $\eta=10^{-2}, \, 10^{-1}$ and $10^{0}.$ As in the previous case, we calculate an $L_2$-norm (see Table \ref{t4}).  
%----------------------------------------------------------------------------------------------------
\begin{table}
  \centering
  \caption{$L_2$ norm. }
  \begin{tabular}{c|c|c|c|c|c|c}\hline
  \multicolumn{7}{c}{Error analysis for fixed \# of nodes over wealth $w.$ }  \\ \hline  
    \multirow{2}{*}{\#  of nodes} & \multicolumn{2}{c|}{ $\eta=0.01$} &
                                                            \multicolumn{2}{c|}{$\eta=0.1$}  & 
                                                                   \multicolumn{2}{c}{$\eta=1.0$}  \\ \cline{2-7}
                  &    $L_2$                     &  $ER $      &    $L_2$   & $ER$ &    $L_2$   & $ER$    \\ \hline
    64/128    &   $1.58$                     &     -         &   $5.02$  & -        &      -     & -      \\ \hline
    128/256  &    $8.4\cdot 10^{-1}$  & $1.88 $   &   $2.56$  &  1.96   & $4.29$ & -       \\ \hline
    256 /512 &    $4.4\cdot 10^{-1}$  & $1.92$    &   $1.29$  &  1.98   & $2.21$ & 1.95   \\ \hline
  \end{tabular}
  \label{t4}
 \end{table}
%---------------------------------------------------------------------------------------------
As we can see from the table, the value of $ER$ for this case is very similar to the previous one for parameter $\eta=0.01$ despite the fact that in this case we chose a different time grid, more sparse then for other  $\eta$ values, namely for $\eta=0.01$ the number of time nodes was approximately $1\cdot 10^3,$ for other $\eta$ values $\sim 4\cdot 10^4.$ As the number of nodes goes up the time that is needed to do computations increases dramatically therefore it becomes more time consuming to perform calculations for very fine grids. For the last case when $\eta=1.0$ some parameters were changed because of stability issues. So, the number of time nodes were increased $~8\cdot 10^{4},$ the number of wealth nodes was decreased to $M=256$ but the final value of the error is good enough.

Overall, based on the $ER$ results  for all three values of the smoothing factor $\eta$ we can see that the error gradually approaches to $2^1$ as expected and we can conclude that our numerical scheme is stable. 
\end{enumerate}
%----------------------------------------------
\section{Numerical results}
\label{numerical_res}
%----------------------------------------------
%----------------------------------------------
\subsection{Numerical solution with pension.}
%----------------------------------------------
In this part we are going to discuss numerical results obtained by solving PDE (\ref{PDE_c_opt}). The parameter set is taken to be the following. Based on papers where authors discuss what is the reasonable value of the risk aversion parameter (e.g. see \cite{Bl}) $\gamma=3,$ risk-free rate $r=0.02,$ volatility $\sigma=0.16,$ and drift $\mu=0.08.$ Moreover, in this paragraph asset allocation is taken to be fixed, $\theta=0.6.$ Then choose pension $\pi=1.$ 
In the next few paragraphs we will discuss different variations of the solution of the existing problem. The goal is to see how optimal consumption changes for the different values of the smoothing parameter $\eta.$
\begin{figure}[H]\centering
\begin{subfigure}[c]{0.4\textwidth}
\includegraphics[scale=0.27]{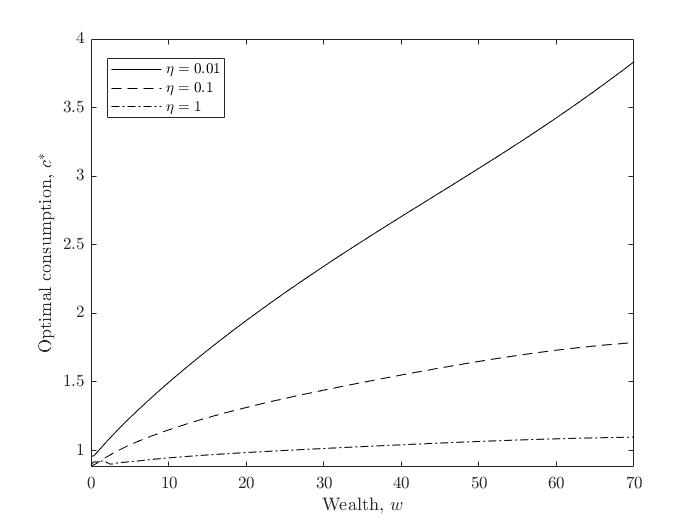}
\caption{$\bar c \approx~ 1.0.$}
\label{f5_1}
\end{subfigure}
\begin{subfigure}[c]{0.4\textwidth}
\includegraphics[scale=0.27]{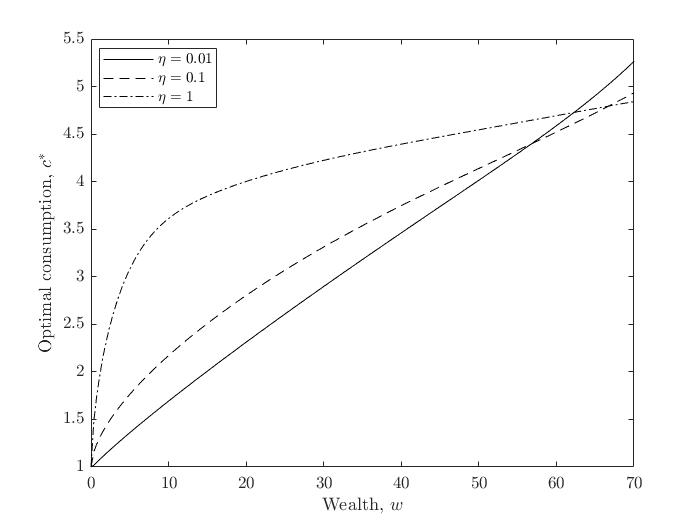}
\caption{$\bar c \approx~ 5.$}
\label{f5_2}
\end{subfigure}
\begin{subfigure}[l]{0.4\textwidth}
\includegraphics[scale=0.27]{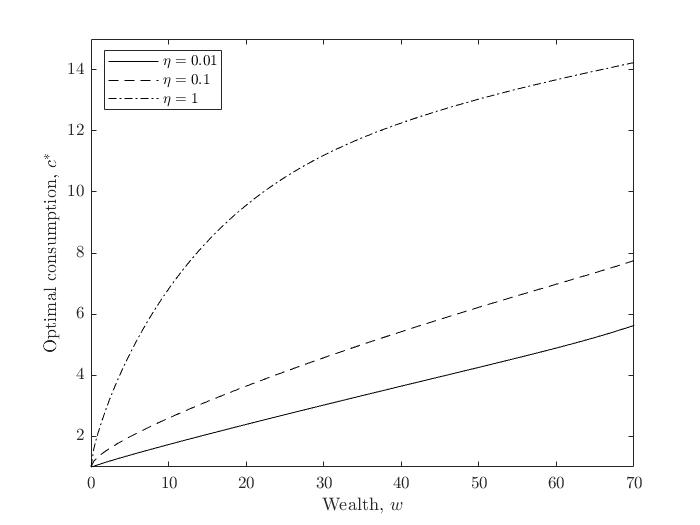}
\caption{$\bar c \approx ~20.$}
\label{f5_3}
\end{subfigure}
\caption{Optimal consumption $c^*$ vs. wealth $ w$ for the fixed value of habit $\bar c$ at age $65$.}
\label{f5}
\end{figure}
%
%---------------------------------------------------------------------------------------
\paragraph{Optimal consumption $c^*$ vs. wealth $w.$}
%---------------------------------------------------------------------------------------
First, we choose three different values of the smoothing factor $\eta$ that characterizes how fast the habit changes, for example $\eta=10^{-2},\, 10^{-1}$ and $1.$ 

Figure \ref{f5} plots the relationship between optimal consumption $c^*$ and wealth $w$ for different values of habit $\bar c$, namely $\bar c=1,\, 5,\, 20.$ 
When $\bar c=1$, consumption stays modest with $\eta=1$, and with $\eta=0.1$ or $0.01,$ consumption is even smaller. With  a higher habit $(\bar c=5)$, the $\eta=1$ curve
stays similar, but the $\eta=0.1$  and $\eta=0.01$  curves cross over, to display higher consumption. With the highest habit $(\bar c=20)$, the $\eta=0.1$ and $\eta=0.01$ consumption rises even
faster with $w$. 

Figure \ref{f6} shows how optimal consumption $c^*$ changes with the wealth $w$ for particular values of the parameter $\eta.$
\begin{figure}[H]\centering
\begin{subfigure}[c]{0.4\textwidth}
\includegraphics[scale=0.27]{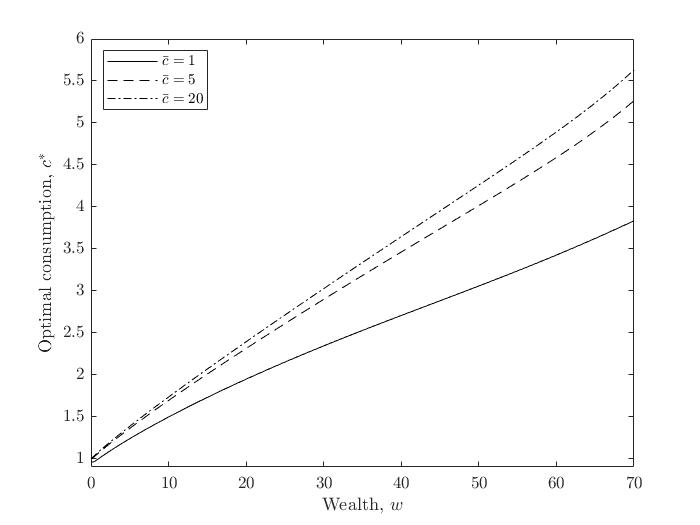}
\caption{$\eta=10^{-2}$}
\label{f6_1}
\end{subfigure}
\begin{subfigure}[c]{0.4\textwidth}
\includegraphics[scale=0.27]{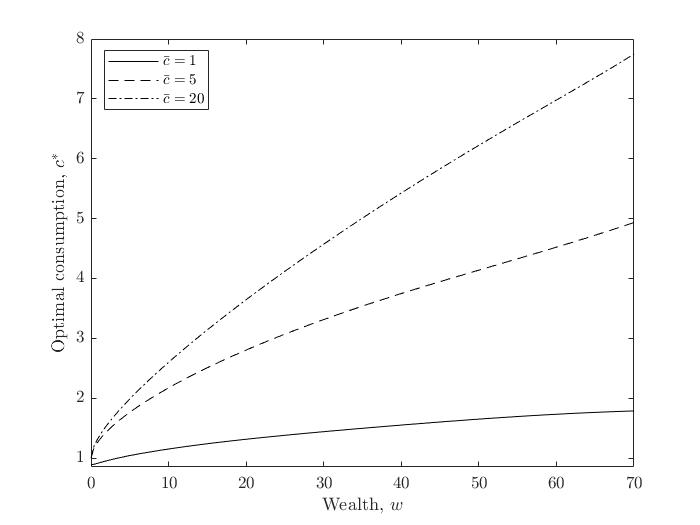}
\caption{$\eta=10^{-1}$}
\label{f6_2}
\end{subfigure}
\begin{subfigure}[c]{0.4\textwidth}
\includegraphics[scale=0.27]{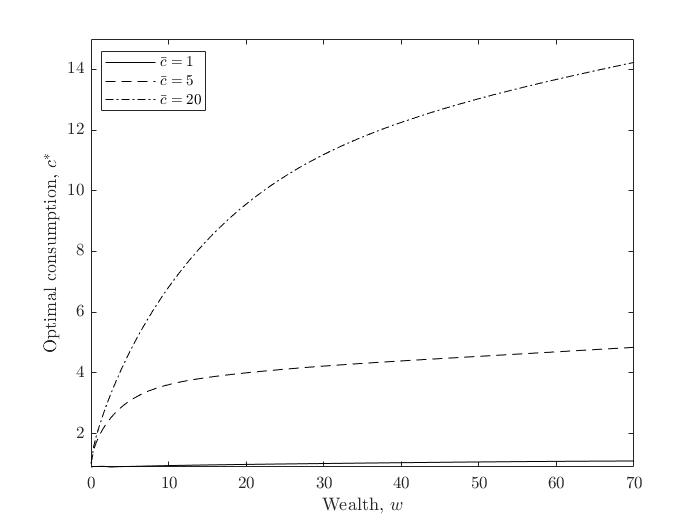}
\caption{$\eta=1$}
\label{f6_3}
\end{subfigure}
\caption{Optimal consumption $c^*$ vs. wealth $ w$  for given $\bar c.$ Every picture represents a different value of $\eta.$ }
\label{f6}
\end{figure}
The variation in consumption with respect to $\bar c$ is the smallest when $\eta$ is small, $\eta=0.01.$ Increasing $\eta$ to $0.1$ widens the spread, with lower consumption with $\bar c=1$ and higher consumption with
$\bar c=20.$ Increasing $\eta$ to $1.0$ widens the spread even further.

Now, let's consider another set of numerical results.
%---------------------------------------------------------------------------------------
\paragraph{Optimal consumption $c^*$ vs. habit $\bar c.$}
%---------------------------------------------------------------------------------------
Figure \ref{f7} shows the relationship between the optimal consumption $c^*$ and habit $\bar c\in[0\dots 30]$ for three different values of the smoothing factor $\eta=0.01,\,0.1,\, 1.$  Every picture corresponds to a fixed value of wealth $w=~1,\,30,\,60.$  For low wealth $(w=1)$,  consumption levels off as $\bar c$ rises, although at different levels. For higher wealth $(w=30),$ consumption still levels off quickly when $\eta=1,$ but it keeps rising for longer when $\eta=0.1$ or $0.01.$ With even higher wealth, this happens even more so.
\begin{figure}[H]\centering 
\begin{subfigure}[c]{0.4\textwidth}
\includegraphics[scale=0.27]{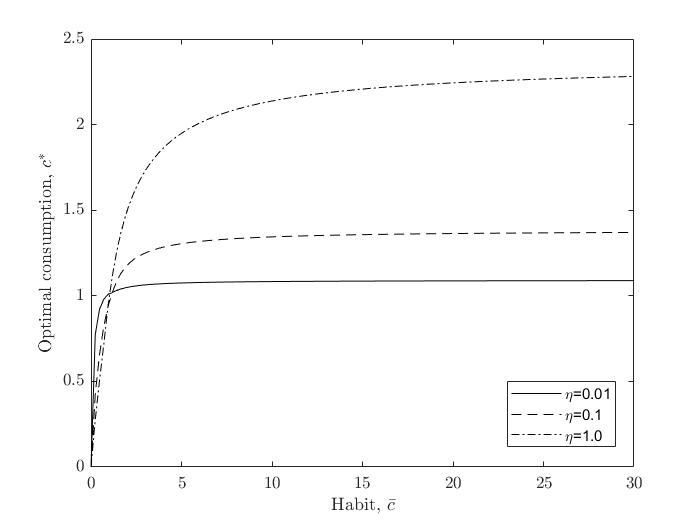}
\caption{$w \approx~ 1.$}
\label{f7_1}
\end{subfigure}
\begin{subfigure}[c]{0.4\textwidth}
\includegraphics[scale=0.27]{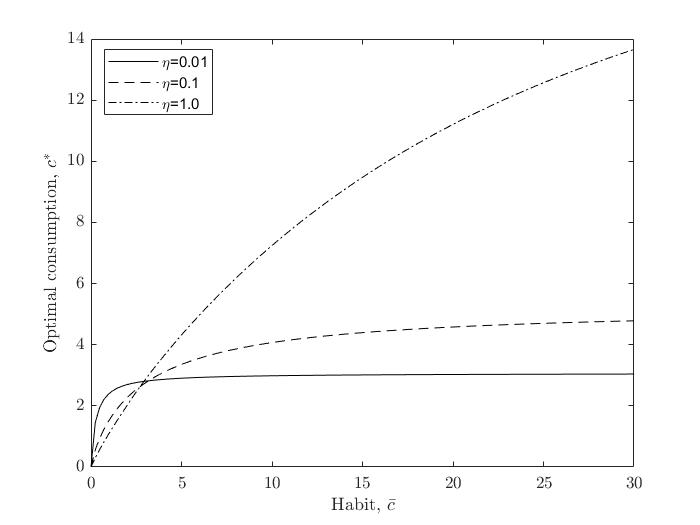}
\caption{$ w \approx ~30.$}
\label{f7_2}
\end{subfigure}
\centering
\begin{subfigure}[c]{0.4\textwidth}
\includegraphics[scale=0.27]{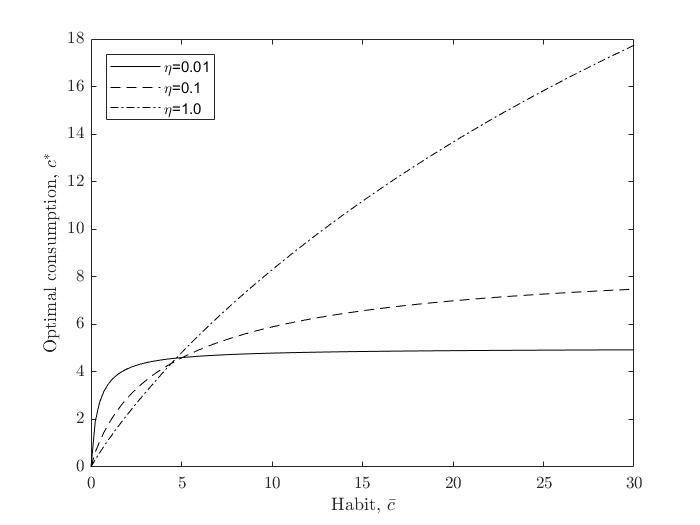}
\caption{$w \approx ~60.$}
\label{f7_4} 
\end{subfigure}
\caption{Optimal consumption $c^*$ vs. habit $\bar c$ where every picture represents a fixed value of wealth.}
\label{f7}
\end{figure}

The last set of pictures (see Figures (\ref{f8_1})-(\ref{f8_3}) shows numerical results similar to the previous set. Each picture corresponds to a fixed value of the smoothing factor $\eta.$ When $\eta=0.01,$ the three curves again level off at modest values of $\bar c.$ For $\eta=0.1$ the $w=1$ curve levels off, while the others continue to rise for longer. And when $\eta=1,$ only the low wealth curve levels off, at least for moderate values of $\bar c.$

\begin{figure}[H]\centering
\begin{subfigure}[c]{0.4\textwidth}
\includegraphics[scale=0.27]{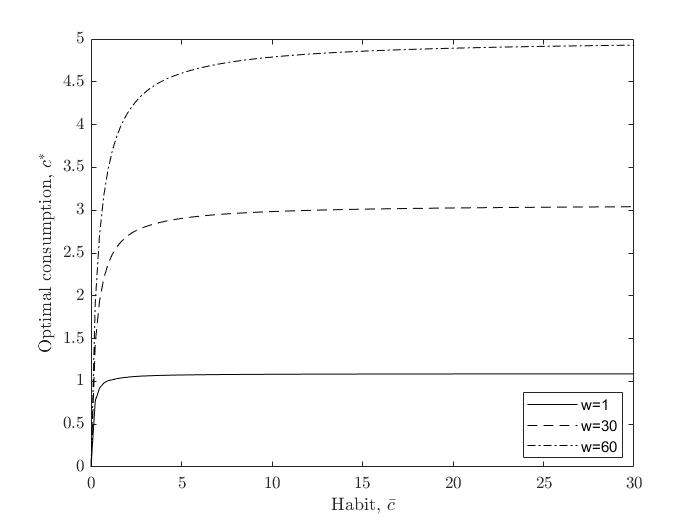}
\caption{$\eta=10^{-2}$}
\label{f8_1}
\end{subfigure}
\begin{subfigure}[c]{0.4\textwidth}
\includegraphics[scale=0.27]{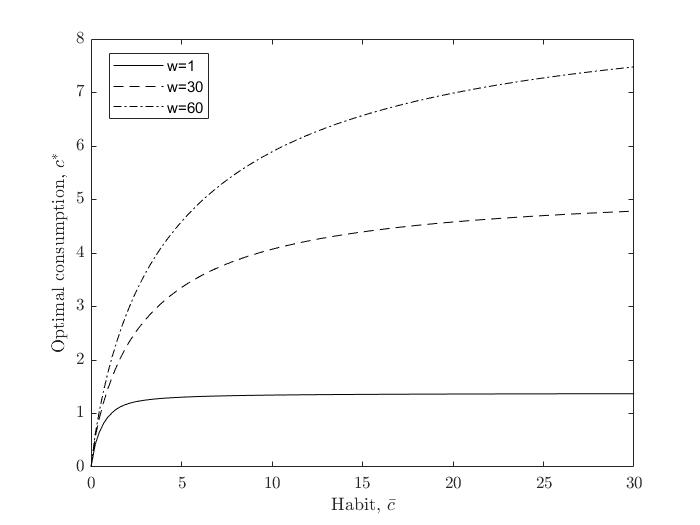}
\caption{$\eta=10^{-1}$}
\label{f8_2}
\end{subfigure}
\begin{subfigure}[c]{0.4\textwidth}
\includegraphics[scale=0.27]{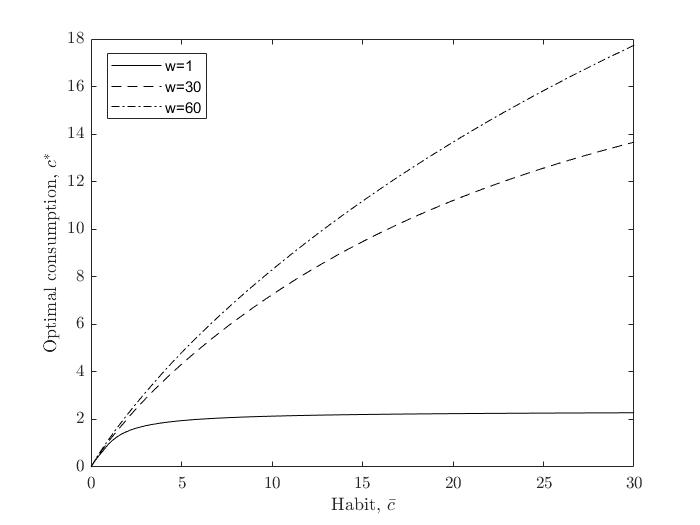}
\caption{$\eta=1$}
\label{f8_3}
\end{subfigure}
\caption{Optimal consumption $c^*$ vs. habit $\bar c$ where every picture represents a fixed value of the smoothing factor $\eta$.}
\label{f8}
\end{figure}

%--------------------------------------------------------------------
\section{Numerical results for different asset allocation values and volatility.}
\label{section_comparison}
%--------------------------------------------------------------------
In this paragraph we explore the impact of varying asset allocation $\theta$ or volatility $\sigma.$ In Figure \ref{f46}, each graph has a fixed value of $\theta$ and shows nine curves, corresponding to three choices
of $\bar c$ and three choices of  $\sigma.$ In Figure \ref{f51} the roles of $\theta$ and $\sigma$ are reversed: each graph has a fixed $\sigma,$ but three choices of $\bar c$ and three choices of $\theta.$

With low $\theta$ (Figure \ref{f48_1}), varying $\sigma$ has little impact on consumption. But for higher $\theta$ (\ref{f48_2} and \ref{f46_1}), the sensitivity to $\sigma$ rises. Though it is still quite small in the case $\bar c=1.$
\begin{figure}[H]\centering
\begin{subfigure}[c]{0.4\textwidth}
\includegraphics[scale=0.27]{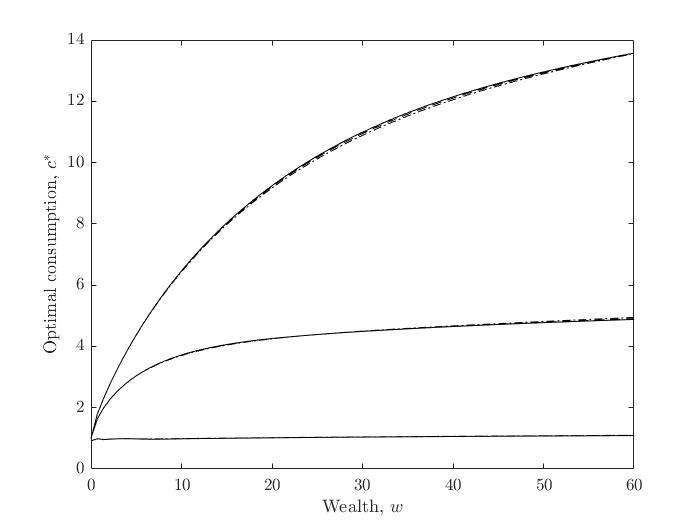}
\subcaption{ $\theta=0.2$ }
\label{f48_1}
\end{subfigure}
\begin{subfigure}[c]{0.4\textwidth}
\includegraphics[scale=0.27]{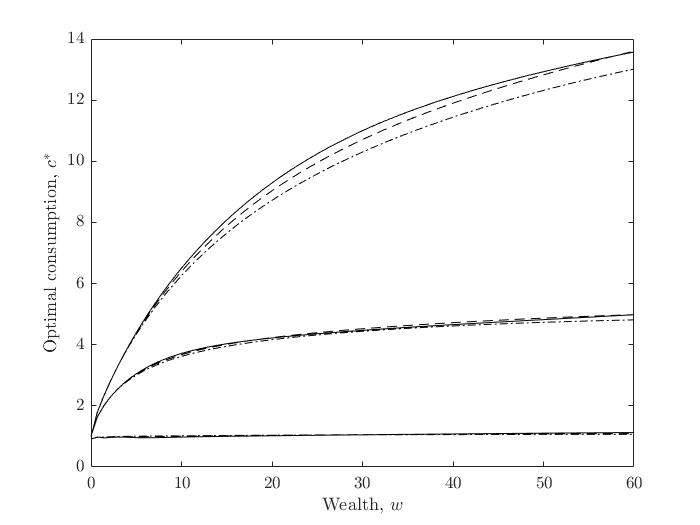}
\subcaption{$\theta=0.6$}
\label{f48_2}
\end{subfigure}
\begin{subfigure}[c]{0.4\textwidth}
\includegraphics[scale=0.27]{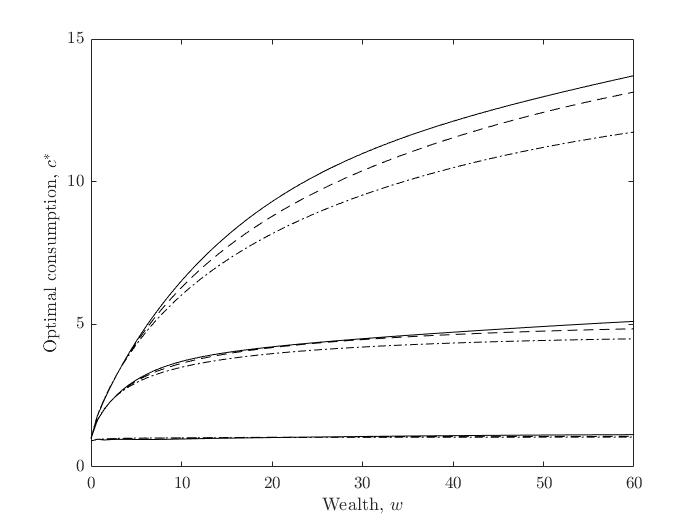}
\subcaption{$\theta=0.9$ }
\label{f46_1}
\end{subfigure}
\caption{Numerical results for OP with pension, for $\eta=1$, , $\bar c=1,\; 5$ and $20$ for fixed value of $\theta=0.9$ for three parameters $\sigma=0.16$ (solid line), $0.50$ (dashed line), $0.75$ (dash-dot line) .}
\label{f46}
\end{figure}
In the same way, with low $\sigma$ (Figure \ref{f49_1}) there is little sensitivity to $\theta,$ but that sensitivity rises when $\sigma$ is large.
\begin{figure}[H]\centering
\begin{subfigure}[c]{0.4\textwidth}
\includegraphics[scale=0.27]{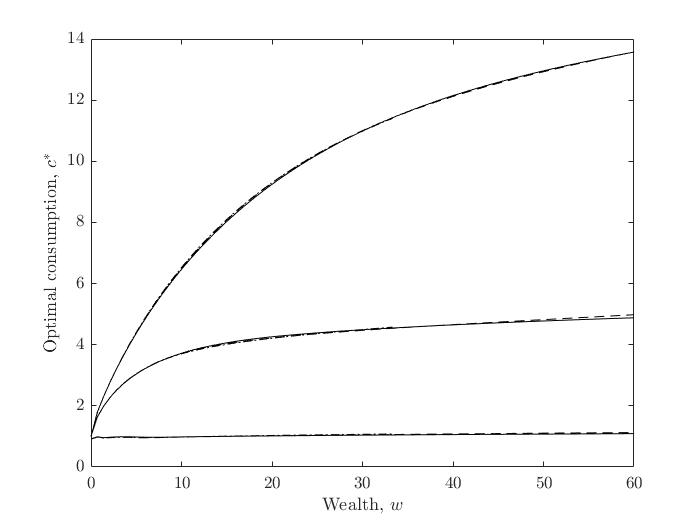}
\subcaption{$\sigma=0.16$ }
\label{f49_1}
\end{subfigure}
\begin{subfigure}[c]{0.4\textwidth}
\includegraphics[scale=0.27]{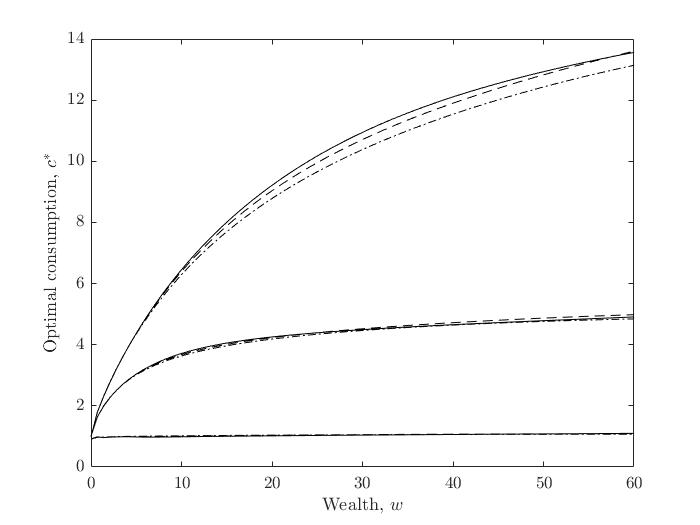}
\subcaption{$\sigma=0.5$ }
\label{f49_2}
\end{subfigure}
\begin{subfigure}[c]{0.4\textwidth}
\includegraphics[scale=0.27]{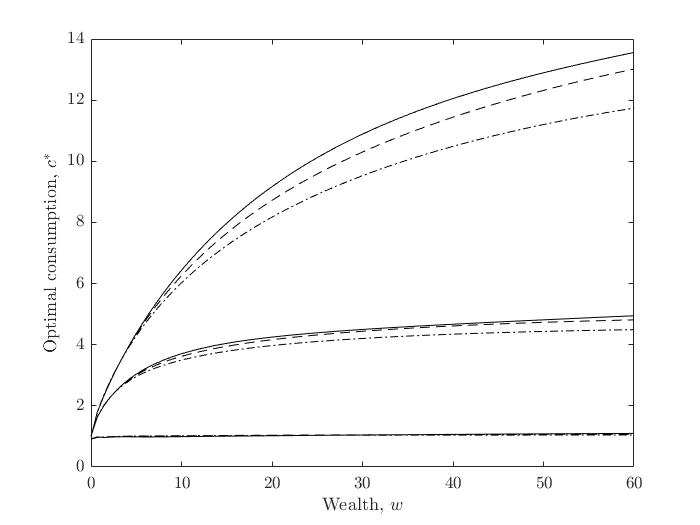}
\subcaption{$\sigma=0.75$ }
\label{f51_1}
\end{subfigure}
\caption{Numerical results for optimization problem with pension, for $\eta=1$,  $\bar c=1,\; 5, \; 20$ and for three parameters $\theta=0.2$ (solid line), $0.6$ (dashed line), $0.9$ (dash-dot line) .}
\label{f51}
\end{figure}
We exhibit this only with $\eta=1.$ For lower choices of $\eta$ we see similar behaviour (not shown), but all the sensitivities increase  when $\eta$ is smaller.
\end{appendix}


\begin{thebibliography}{15}
%-----------------------------------------------
%\bibitem{B}{\sc Z. Bodie, J.B.Detemple, S. Otruba, S. Walter}, {\em Optimal consumption-portfolio choices and retirement planning}, J. Econom. Dynam. Control, 28 (2004),  pp.~1115 -- 1148.
%\bibitem{Bl}{\sc D. Blake, and W. Burrows}, {\em Survivor bonds: helping to hedge mortality risk}, Journal of Risk and Insurance 68 (2001), pp.~339--348.
%\bibitem{Br} {\sc J.R. Brown, O.S. Mitchell, J.M. Poterba, and M.J. Warshawsky}, {\em The role of annuity markets in financing retirement}, Cambridge: MIT Press, 2001.
%\bibitem{CA} {\sc C.D. Carroll, J. Overland and D.N. Weil}, {\em Saving and growth with habit formation}, The American Economic Review,  90(3) (2000), pp.~ 341--355.
%\bibitem{CH} {\sc R. Chetty and  A. Szeidl}, {\em Consumption commitments and habit formation}, Econometrica, 84(2) (2016), pp.~855--890.
%\bibitem{C} {\sc George Constantinides}, {\em Habit formation: a resolution of the equity premium puzzle}, The Journal of Political Economy, 98(3) (1990), pp.~ 519--543.
%\bibitem{Jer}{\sc J. Detemple, F. Zapatero}, {\em Optimal consumption-portfolio policies with habit formation}, Math. Finance, 2(4) (1992), pp.~251--274.
%\bibitem{D}{\sc R. Durrett}, {\em Probability: theory and examples}, Cambridge University Press, 2013.
%\bibitem{H}{\sc F. Habib, H. Huang and M.A. Milevsky}, {\em Approximate solutions to retirement spending problems and the optimality of ruin}, Pension Risk Management eJournal, 2017.
%\bibitem{HH}{\sc H. Huang, M.A. Milevsky and T.S. Salisbury}, {\em Retirement spending and biological age},  J. Econom. Dynam. Control, 84 (2017), pp.~58--76.
%\bibitem{MP}{\sc R. Mehra, Ed.C. Prescott}, {\em The equity premium: a puzzle}, J. Monetary Econom., 15(2) (1985), pp.~145--161.
%\bibitem{M}{\sc M.A. Milevsky}, {\em The calculus of retirement income: financial models for pension annuities and life insurance}, Cambridge, 2006.
%\bibitem{Mil}{\sc M.A. Milevsky}, {\em Swimming with wealthy sharks: longevity, volatility and the value of risk pooling}, J. Pension Econom. Finance,  Vol. 19, Issue 2 (2020), 217-246.
%\bibitem{MHH}{\sc M.A. Milevsky and H. Huang}, {\em Spending retirement on planet vulcan: the impact of longevity risk aversion on optimal withdrawal rates}, Risk Management, 23 (2011), pp.~24--38.
%\bibitem{MH}{\sc M.A. Milevsky and H. Huang}, {\em The utility value of longevity risk pooling: analytic insights}, North American Actuarial J. (2019), pp.~574--590.
%\bibitem{N}{\sc R. Naryshkin and M. Davison}, {\em Numerical methods for portfolio optimization under habit formation and transaction costs}, Proc. China-Canada Industry Workshop on Enterprise Risk Management (2008).
%\bibitem{P}{\sc V. Polkovnichenko}, {\em Life-cycle portfolio choice with additive habit formation preferences and uninsurable labor income risk}, Rev. of Financial Studies, 20(1) (2007), pp.~83--124.
%\bibitem{Pollak}{\sc R.A. Pollak}, {\em Habit formation and dynamic demand functions}, J. of Political Economy, 78(4), part 1 (1970), pp.~745--763.
%\bibitem{Rei}{\sc F. Reichling and K. Smetters}, {\em Optimal annuitization with stochastic mortality and correlated mortality cost}, Amer. Econom. Rev., 11 (2015), pp.~3273--3320.
%\bibitem{R}{\sc L.C.G. Rogers}, {\em Optimal investment}, Springer, 2013.
%\bibitem {S}{\sc J. Strikwerda}, {\em Finite difference schemes and partial differential equations}, SIAM, 2004.
%\bibitem{X} {\sc R. Xinfeng, W. Zhu, J. Hu, J. Huang}, {\em Optimal portfolio and consumption with habit formationin a jump diffusion market}, Appl. Math. Comput., 222 (2013), pp.~391--401.
\bibitem{B}{\textsc{Z. Bodie, J.B.Detemple, S. Otruba, S. Walter}}, {\em Optimal consumption-portfolio choices and retirement planning}, J. Econom. Dynam. Control, 28 (2004),  pp.~1115 -- 1148.
\bibitem{Bl}{\textsc{ D. Blake, and W. Burrows}}, {\em Survivor bonds: helping to hedge mortality risk}, Journal of Risk and Insurance 68 (2001), pp.~339--348.
\bibitem{Br} {\textsc{ J.R. Brown, O.S. Mitchell, J.M. Poterba, and M.J. Warshawsky}}, {\em The role of annuity markets in financing retirement}, Cambridge: MIT Press, 2001.
\bibitem{CA} {\textsc{ C.D. Carroll, J. Overland and D.N. Weil}}, {\em Saving and growth with habit formation}, The American Economic Review,  90(3) (2000), pp.~ 341--355.
\bibitem{CH} {\textsc{ R. Chetty and  A. Szeidl}}, {\em Consumption commitments and habit formation}, Econometrica, 84(2) (2016), pp.~855--890.
\bibitem{C} {\textsc{ George Constantinides}}, {\em Habit formation: a resolution of the equity premium puzzle}, The Journal of Political Economy, 98(3) (1990), pp.~ 519--543.
\bibitem{Jer}{\textsc{ J. Detemple, F. Zapatero}}, {\em Optimal consumption-portfolio policies with habit formation}, Math. Finance, 2(4) (1992), pp.~251--274.
\bibitem{D}{\textsc{ R. Durrett}}, {\em Probability: theory and examples}, Cambridge University Press, 2013.
\bibitem{H}{\textsc{ F. Habib, H. Huang and M.A. Milevsky}}, {\em Approximate solutions to retirement spending problems and the optimality of ruin}, Pension Risk Management eJournal, 2017.
\bibitem{HH}{\textsc{ H. Huang, M.A. Milevsky and T.S. Salisbury}}, {\em Retirement spending and biological age},  J. Econom. Dynam. Control, 84 (2017), pp.~58--76.
\bibitem{MP}{\textsc{ R. Mehra, Ed.C. Prescott}}, {\em The equity premium: a puzzle}, J. Monetary Econom., 15(2) (1985), pp.~145--161.
\bibitem{M}{\textsc{ M.A. Milevsky}}, {\em The calculus of retirement income: financial models for pension annuities and life insurance}, Cambridge, 2006.
\bibitem{Mil}{\textsc{ M.A. Milevsky}}, {\em Swimming with wealthy sharks: longevity, volatility and the value of risk pooling}, J. Pension Econom. Finance,  Vol. 19, Issue 2 (2020), 217-246.
\bibitem{MHH}{\textsc{ M.A. Milevsky and H. Huang}}, {\em Spending retirement on planet vulcan: the impact of longevity risk aversion on optimal withdrawal rates}, Risk Management, 23 (2011), pp.~24--38.
\bibitem{MH}{\textsc{ M.A. Milevsky and H. Huang}}, {\em The utility value of longevity risk pooling: analytic insights}, North American Actuarial J. (2019), pp.~574--590.
\bibitem{N}{\textsc{ R. Naryshkin and M. Davison}}, {\em Numerical methods for portfolio optimization under habit formation and transaction costs}, Proc. China-Canada Industry Workshop on Enterprise Risk Management (2008).
\bibitem{P}{\textsc{ V. Polkovnichenko}}, {\em Life-cycle portfolio choice with additive habit formation preferences and uninsurable labor income risk}, Rev. of Financial Studies, 20(1) (2007), pp.~83--124.
\bibitem{Pollak}{\textsc{ R.A. Pollak}}, {\em Habit formation and dynamic demand functions}, J. of Political Economy, 78(4), part 1 (1970), pp.~745--763.
\bibitem{Rei}{\textsc{ F. Reichling and K. Smetters}}, {\em Optimal annuitization with stochastic mortality and correlated mortality cost}, Amer. Econom. Rev., 11 (2015), pp.~3273--3320.
\bibitem{R}{\textsc{ L.C.G. Rogers}}, {\em Optimal investment}, Springer, 2013.
\bibitem {S}{\textsc{ J. Strikwerda}}, {\em Finite difference schemes and partial differential equations}, SIAM, 2004.
\bibitem{X} {\textsc{ R. Xinfeng, W. Zhu, J. Hu, J. Huang}}, {\em Optimal portfolio and consumption with habit formationin a jump diffusion market}, Appl. Math. Comput., 222 (2013), pp.~391--401.
\end{thebibliography}
\end{document}